\documentclass[10pt]{iopart}
\usepackage{amssymb}
\usepackage{amsthm}
\usepackage{amscd}
\usepackage[usenames,dvipsnames]{color}
\usepackage{pstricks}
\usepackage{graphicx}

\newcommand{\dotT}{\dot{T}}

\newcommand{\telque}{{\,;\,}}

\newcommand{\bbR}{{\mathbb{R}}}

\newcommand{\bfk}{{\mathbf{k}}}

\newcommand{\bfx}{{\mathbf{x}}}

\newcommand{\calD}{{\mathcal{D}}}
\newcommand{\calE}{{\mathcal{E}}}
\newcommand{\calF}{{\mathcal{F}}}

\newcommand{\cotg}{{\mathrm{cotg}}}

\newcommand{\supp}{{\mathrm{supp\,}}}
\newcommand{\singsupp}{{\mathrm{sing\,supp\,}}}
\newcommand{\dd}{{{d}}}

\newcommand{\WF}{{\mathrm{WF}}}

\newcommand{\caplus}{\oplus}

\newcommand{\text}[1]{\,\,\mathrm{#1}\,\,}
\newcommand{\eqref}[1]{(\ref{#1})}

\newtheorem{prop}{Proposition}
\newtheorem{thm}[prop]{Theorem}

\newtheorem{example}[prop]{Example}

\newtheorem{dfn}[prop]{Definition}

\begin{document}
%
%
\title[Introduction to the wavefront set]{A smooth introduction 
to the wavefront set}
\author{Christian Brouder}
\address{Institut de Min\'eralogie, de Physique des Mat\'eriaux et de
Cosmochimie, Sorbonne Universit\'es,
UMR CNRS 7590, UPMC Univ. Paris 06, Mus\'eum National d'Histoire
Naturelle, IRD UMR 206,
4 place Jussieu, F-75005 Paris, France.}
\author{Nguyen Viet Dang and Fr\'ed\'eric H\'elein}
\address{Institut de Math\'ematiques de Jussieu - Paris
  Rive Gauche, \\Universit\'e Paris Diderot-Paris 7,
  CNRS UMR7586, Case 7012, B\^at. Sophie Germain,
  75205 Paris cedex 13, France.}
\begin{abstract}
The wavefront set provides a precise description of the
singularities of a distribution. Because of its ability
to control the product of distributions, the wavefront
set was a key element of recent progress in
renormalized quantum field theory in curved spacetime,
quantum gravity, the discussion of time machines
or quantum energy inequalitites.
However, the wavefront set is a somewhat subtle concept
whose standard definition is not easy to grasp.
This paper is a step by step introduction to the
wavefront set, with examples and motivation.
Many different definitions and new interpretations
of the wavefront set are presented. Some of them
involve a Radon transform.
\end{abstract}

\maketitle

\section{Introduction}

\label{mul-dis-sect}
Feynman propagators are distributions, and
Stueckelberg realized very early that renormalization
was essentially the problem of defining a product
of distributions~\cite{Rivier,StueckelbergR,StueckelbergG}.
This point of view was clarified by Bogoliubov, Shirkov,
Epstein and Glaser~\cite{BS56,Bogoliubov,Epstein}
but was almost forgotten.

In a ground-breaking paper~\cite{Radzikowski}, Radzikowski showed that
the \emph{wavefront set} of a distribution was a crucial
concept to define quantum fields in curved spacetime.
This idea was fully developed into a renormalized
scalar field theory in curved spacetimes by 
Brunetti, Fredenhagen~\cite{Brunetti2}, 
Hollands and Wald~\cite{Hollands2}.
This approach was rapidly extended to the case of
Dirac fields~\cite{Kratzert-00,Hollands-01,Antoni-06,%
Dappiaggi-09,Sanders-10-Dirac,Rejzner-11,Zahn-14},
to gauge fields~\cite{Hollands-08,Fredenhagen-11,Fredenhagen-13}
and even to the quantization of gravitation~\cite{Brunetti-13-QG}.

This tremendous progress was made possible by
a complete reformulation of quantum field theory,
where the wavefront set of distributions plays a
central role, for example to determine the algebra
of microcausal functions and to define a 
spectral condition for time-ordered 
products and 
quantum states~\cite{BFK,Verch-99,Hollands3,Sanders-10}.
The wavefront set was also a decisive tool to discuss the existence
of time-machine spacetimes~\cite{Kay-97}, quantum
energy inequalities~\cite{Fewster-00} and 
cosmological models~\cite{Pinamonti-11}.

Until the early 90s, the wavefront set was rarely
used to solve physical problems.
We only know of a few works in crystal
optics~\cite{Ivrii-77,Esser-87}
and quantum field theory on curved
spacetimes~\cite{Moreno-77,Dimock-79}.
This is probably due to the fact that this concept
is not familiar to most physicists and not easy
to grasp. But now, the wavefront set is here to 
stay and we think that a smooth and physically motivated
introduction to it is worthwhile. This is the purpose
of the present paper.

There are textbook descriptions of the wavefront
set~\cite{HormanderI,Duistermaat,Chazarain,ReedSimonII,%
Friedlander,Strichartz-03,Grigis,Strohmaier,Eskin,Wagschal-11},
but they do not give any clue on its physical meaning
and advanced textbooks are notoriously laconic
(the outstanding exception 
being the book by Gregory Eskin~\cite{Eskin}).

The main use of the wavefront set in quantum field theory is
to provide a condition for the product of distributions.
Indeed, the Feynman propagator is a distribution and
the products of propagators present in a Feynman diagrams
are not well defined. The wavefront set gives a precise 
description of the region of spacetime where the product
is well defined and the value of the Feynman diagram on
the whole spacetime is then obtained by an extension 
procedure~\cite{Brunetti2}.

After this introduction, we discuss in simple terms the
problem of the multiplication of one-dimensional distributions.
This elementary example reveals a natural condition for
two distributions to be multiplied and this condition leads
to the definition of the wavefront set.
After giving elementary examples of wavefront sets, we discuss
in detail the wavefront set of the characteristic function
of a domain $\Omega$ in the plane
(i.e. a function which is equal to 1 on $\Omega$
to 0 outside it).
To bring a physical feel of the concept, we
give two new characterizations of the wavefront set:
the first one uses a Radon transform, the second one counts
the number of intersections of straight lines with the boundary 
of $\Omega$. These two characterizations do not
employ any Fourier transform.
The next section explores the wavefront set of a distribution
defined by an oscillatory integral. This technique is crucial
to calculate the wavefront set of the Wightman and Feynman
propagators in quantum field theory.
The main properties of the wavefront set are listed without proof.
The last section enumerates other definitions of
the wavefront set. 

\section{Multiplication of distributions}
\label{multdissect}
We shall introduce the wavefront set as a condition
required to multiply distributions. We first recall that
a distribution $u\in\calD'(\bbR^n)$
is a continuous linear map from the set of smooth compactly supported
functions $\calD(\bbR^n)$ to the complex numbers,
and we denote $u(f)$ by $\langle u,f\rangle$. For example,
if $\delta$ is the Dirac delta distribution, then
$\langle \delta,f\rangle=f(0)$. If $g$ is a locally integrable
function, then we can consider it as a distribution
by associating to $g$ the distribution
$\langle u_g,f\rangle=\int g(x) f(x) dx$
(for a nice introduction to distributions see for example
\cite{Friedlander}).

It is well known that distributions can generally not
be multiplied~\cite{Schwartz}. 
The first reason is the very
definition of distributions as objects which generalize the functions
but
for which the `value at some point' has no sense in general. But,
motivated by
questions in theoretical physics (e.g. quantum field theory),
we may ask under which circumstances it is possible
to \emph{extend} the product of ordinary functions to distributions.
In most cases this is just impossible.
For instance we cannot make sense of the
square of $\delta$: 
a simple way to convince yourself of that is to study the family
of functions $\chi_\varepsilon:\mathbb{R}\longrightarrow \mathbb{R}$ for
$\varepsilon>0$
defined by $\chi_\varepsilon(x) =1/\varepsilon$
if $|x|\leq \varepsilon/2$ and $\chi_\varepsilon(x) = 0$ otherwise.
For any $f\in\calD(\bbR)$ we have
$\int_\bbR \chi_\varepsilon(x) f(x) dx = 
\varepsilon^{-1} \int_{-\varepsilon/2}^{\varepsilon/2}
f(x) dx = \varepsilon^{-1} (\varepsilon f(0) + O(\varepsilon^3))$
and $\lim_{\varepsilon\to0} \chi_\varepsilon = \delta$.
However, the square of $\chi_\varepsilon$ does not 
converge to a distribution: 
$\int_\bbR \chi^2_\varepsilon(x) f(x) dx = 
\varepsilon^{-2} \int_{-\varepsilon/2}^{\varepsilon/2}
f(x) dx = \varepsilon^{-2} (\varepsilon f(0) + O(\varepsilon^3))$
diverges for $\varepsilon\to 0$.

In some other cases it is possible to define a product, but we
loose some good properties.
Consider the example of the
Heaviside step function $H$, which is defined
by $H(x)=0$ for $x <0$ and $H(x)=1$ for
$x\ge 0$. Its associated 
distribution, denoted by $\theta$, is 
\begin{eqnarray*}
\langle \theta,f\rangle &=&
\int_{-\infty}^\infty H(x)f(x) dx =  \int_0^\infty f(x) dx.
\end{eqnarray*}
The \emph{function} $H$ can obviously be multiplied with itself
and $H^n=H$ for any integer $n>0$.
As we shall see, it is possible to define a product of distributions
such that $\theta^n=\theta$ as a distribution.
But then, we loose the compatibility of the product
with the Leibniz rule because, by taking 
the derivative of both sides we
would obtain $n\theta^{n-1}\theta'=\theta'$.
The identity $\theta'=\delta$ and
$\theta^{n-1}=\theta$
would give us
$n\theta\delta=\delta$ for all integers $n>1$.
Since the left hand side depends linearly on $n$ and the right
hand side does not and is not equal to zero,
we reach a contradiction.

The Leibniz rule is essential for applications in 
mathematical physics and we shall
define a product of distributions 
obeying the Leibniz rule. 
We first enumerate some conditions under which
distributions can be safely multiplied.

\subsection{In which cases can we multiply distributions ?
}\label{inwhichcases}
\subsubsection{A distribution times a smooth function}
The product of distributions is well
defined when one of the two distributions is a smooth
function. Indeed, consider a distribution
$u\in \calD'(\bbR^n)$ and a smooth function
$\phi\in C^\infty(\bbR^n)$. Then, for all test function
$f\in\calD(\bbR^n)$ we
can define the product of $u$ and $\phi$ by
$\langle u\phi,f\rangle=\langle u,\phi f\rangle$.

\subsubsection{Distributions with disjoints singular supports}
We can also define the product of two distributions
when the singularities of the distributions are disjoint.
To make this more precise, we recall that the \emph{support
of a function} $f$, denoted by $\supp f$, is
the closure of the set of points where the
function is not zero~\cite[p.~14]{HormanderI}. 
For example, the support of the Heaviside function is
$\supp H={[}0,+\infty{[}$. Note that although a function is
zero outside its support, it can also vanish
at isolated points of its support, because of the closure
condition
of the definition. For example the support of the sine function
is $\bbR$ although $\sin (n\pi)=0$.

However, the \emph{support of a distribution} cannot
be defined as the support
of a function because the value of a distribution
at a point is generally not defined. Hence we
define the support by duality:
we say that
the point $x$ does not belong to the support of the distribution
$u$ if and only if there is an open neighborhood $U$ of $x$ such that
$u$ is zero on $U$, in other words if
$\langle u,f\rangle=0$ for all test functions $f$ whose
support is contained in $U$~\cite[p.~12]{Friedlander}.
For example $\supp \delta = \{0\}$
and $\supp\theta=[0,+\infty]$.
Similarly, we can define the singular support of a distribution
$u\in\calD'(\bbR^n)$, denoted by $\singsupp u$, by saying that
$x\notin\singsupp u$ if and only if there is a neighborhood
$U$ of $x$ such that the restriction of $u$ to $U$ is
a smooth function, in other words if there is a smooth
function $\phi\in C^\infty(U)$ such that
$\langle u,f\rangle=\langle \phi,f\rangle=\int \phi(x)f(x) dx$
for all test functions $f$ supported on 
$U$~\cite[p.~108]{Friedlander}.
For example $\singsupp\delta=\{0\}$, 
$\singsupp\theta=\{0\}$.

A more elaborate example is the distribution $u\in
\mathcal{D}'(\mathbb{R})$,
defined by:
$u(x)=(x+i0^+)^{-1}$, i.e. $u$
is the limit
in $\mathcal{D}'(\mathbb{R})$ of $u_\varepsilon(x):=
(x+i\epsilon)^{-1}$, when
$\varepsilon >0$ and $\varepsilon\rightarrow 0$, this means
that~\cite[\S~2]{Gelfand-ShilovI}:
\begin{eqnarray*}
\langle u, f\rangle &=& \lim_{\epsilon\to 0^+} 
\int_{-\infty}^\infty \frac{f(x) dx}{x+i\epsilon}
= \lim_{\epsilon\to 0^+} \int_{\epsilon}^\infty
  \frac{f(x)-f(-x)}{x}dx -i \pi f(0). 
\end{eqnarray*}
If $y\not=0$, consider the open set 
$U=(y-|y|/2,y+|y|/2)$. 
Take a smooth function $\chi$ such that $\chi(x)=1$ 
for $|x-y|< 3|y|/4$ and
$\chi(x)=0$ for $|x-y|> 7|y|/8$.
Then, for any $f$ supported on $U$ we have $f(0)=0$
and $f=f\chi$. Thus,
\begin{eqnarray*}
\langle u, f\rangle &=& \langle u,\chi f\rangle
= \int_{|y|/8}^\infty
  \frac{\chi(x)f(x)-\chi(-x)f(-x)}{x}dx 
\\&=&
 \int_{-\infty}^\infty
  \frac{\chi(x)f(x)}{x}dx =
  \langle \phi,f\rangle,
\end{eqnarray*}
where $\phi(x)=\chi(x)/x$ is smooth because $\chi(x)=0$
for $|x|<|y|/8$ (see fig.~\ref{figchi}).
\begin{figure}
\begin{center}
\includegraphics[width=7.0cm]{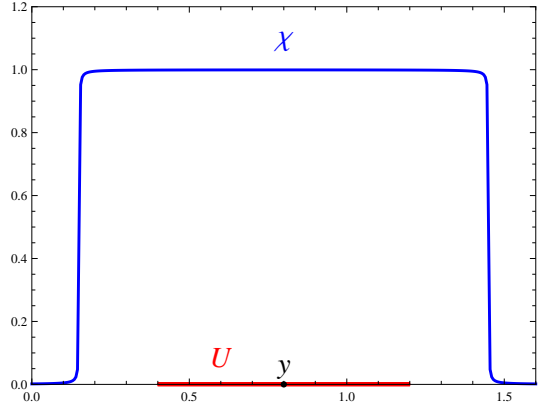}
\caption{In this figure we take $y=0.8$, the open set is
$U=(0.4,1.2)$ and the smooth function $\chi$
is supported on $(0.1,1.5)$.
  \label{figchi}}
\end{center}
\end{figure}
As a consequence, every $y\not=0$ is
not in the singular spectrum of $u$ 
and $\singsupp u=\{0\}$ because the imaginary
part of $u$ is proportional to a Dirac $\delta$ distribution.

We can now state an important theorem~\cite[p.~55]{HormanderI}.
\begin{thm}
\label{singsupthm}
If $u$ and $v$ are two distributions in $\calD'(\bbR^n)$ such
that $\singsupp u \cap \singsupp v=\emptyset$, then the
product $uv$ is well defined.
\end{thm}
\begin{proof}
We first notice that, if $f\in \calD(\bbR^n)$ is supported outside
the singular support of $v$, then $vf$ is smooth 
and we can define the product by
$\langle uv,f\rangle=\langle u,vf\rangle$.
Similary, $\langle uv,f\rangle=\langle v,uf\rangle$
if $f$ is supported outside the singular support of $u$.
This definition of $uv$ extends to all test 
functions $f$
by using a smooth function $\chi$  which is equal to zero on a neighborhood of 
the singular support of $u$ and equal to one on a neighborhood of the singular
support of $v$. Then
$\langle uv,f\rangle=\langle v,u\chi f\rangle + 
\langle u,v (1-\chi) f\rangle$.
This product is associative and commutative~\cite[p.~55]{HormanderI}.
\end{proof}

\subsubsection{The singular oscillations of the distributions are
transversal}
Consider the two distributions $u= \delta \otimes 1$ and $v = 1\otimes
\delta$ in $\mathcal{D}'(\mathbb{R}^2)$,
i.e., $\forall \varphi \in \mathcal{D}(\mathbb{R}^2)$, $\langle
u,\varphi\rangle = \int_\mathbb{R}\varphi(0,y)dy$
and $\langle v,\varphi\rangle = \int_\mathbb{R}\varphi(x,0)dx$. Then we
can define their product by
$uv = (\delta \otimes 1)(1\otimes \delta):= \delta\otimes \delta = 
\delta^{(2)}$, i.e.
$\langle uv,\varphi\rangle=\varphi(0,0)$,
since $\langle uv,\varphi\rangle=\int\int u(x)v(y)\varphi(x,y)dxdy=\int u(x) \left(\int v(y)\varphi(x,y)dy\right) dx=\int u(x) \varphi(x,0)dx=\varphi(0,0)$
by the Fubini theorem for distributions.
Here $u$ and $v$ are singular on the lines $\{x=0\}$ and $\{y=0\}$
respectively, which have a non empty
intersection $\{(0,0)\}$. However the oscillations of both distributions
are orthogonal at that
point, so that this definition makes sense.
But actually the orthogonality is not essential and, as we will see, the important
point is the transversality.

Indeed we can extend this example to measures which are supported by non orthogonal
lines: let
$\alpha:\mathbb{R}^2\longrightarrow \mathbb{R}^2$ be a \emph{linear invertible} map and set
$\alpha = (\alpha^1,\alpha^2)$
and $u_\alpha: = \alpha^*u = u\circ \alpha$ and $v_\alpha := \alpha^*v =
v\circ \alpha$, where $\forall w\in \mathcal{D}'(\mathbb{R}^2)$, $\forall\varphi\in
\mathcal{D}(\mathbb{R}^2)$,
$\langle \alpha^*w,\varphi \rangle:= (\hbox{det}\alpha)^{-1}\langle
w,\varphi \circ \alpha^{-1} \rangle$.
These distributions are well-defined 
and they are singular on the line of equation $\alpha^1=0$ and
$\alpha^2=0$ respectively. Moreover we can define
$u_\alpha v_\alpha$
by setting $u_\alpha v_\alpha:= \alpha^*(uv) =
\alpha^*(\delta^{(2)})$.
Hence here $u_\alpha v_\alpha =
(\hbox{det}\alpha)^{-1}\delta^{(2)}$ and
we see that the product makes sense as long as $\hbox{det}\alpha\neq 0$,
which means that
the singular supports of $u_\alpha$ and $v_\alpha$ are
\emph{transversal}.

\subsubsection{The singularities of the distributions are transversal in the
complex world}
This last case looks as the most mysterious at first glance and concerns
complex
valued distributions. Consider the distribution $u(x) = 1/(x+i0^+)$
defined 
previously, i.e. the limit of $u_\epsilon(x)=1/(x+i\varepsilon) =
\frac{x}{x^2+\varepsilon^2} - \frac{i\varepsilon}{x^2+\varepsilon^2}$
when $\varepsilon>0$ and $\varepsilon\rightarrow 0$.
(Hence $u = \mathrm{pv}(\frac{1}{x}) -i\pi \delta_0$.)
Observe that $(u_\varepsilon)' = - (u_\varepsilon)^2$, $\forall
\varepsilon>0$. Thus since
$(u_\varepsilon)'$ converges to $u'$ in $\mathcal{D}'(\mathbb{R})$, we
can set 
$u^2 := -u'$. Moreover since any polynomial relation in $u_\varepsilon$
and its derivatives
which follows from Leibniz rule is satisfied ($u_\varepsilon$ being a
smooth function),
the same holds for $u$. One can define similarly the square of
$\overline{u}(x) =  1/(x-i0^+)$. However this
recipe fails for defining the product of $u$ by $\overline{u}$.

A similar mechanism works for making sense of the square of the
\emph{Wightman function} (see Section \ref{oscillatory}).
One way to understand what's happening is to remark that we multiply
distributions which are
boundary values of holomorphic functions on the same domain.

In order to really understand all these examples and go beyond, we need
to revisit
them by using refined tools such as: the Radon transform and
the Fourier transform. This will lead us to H{\"o}rmander's 
definition of wavefront sets.

\subsection{The product of distributions by using Fourier transform}
We remark that the Fourier transform
of a product of distributions (when it is defined) is the convolution
of the Fourier transforms of these
distributions~\cite[p.~102]{Friedlander}:
$\widehat{uv}=\hat{u}\star \hat{v}$, if it exists.
Therefore, we can define the product of two distributions
$u$ and $v$ as the inverse Fourier transform of $\hat{u}\star \hat{v}$.
However, this definition, which requires the Fourier transforms of $u$ and $v$
to be defined and their convolution product to make sense, can be improved. Indeed
it does not take into account 
the fact that the product of two distributions is \emph{local},
i.e. that its definition on the neighbhorhood of a point depends only on
the restriction of the distributions on that neighborhood.
Therefore, we can localize the distributions by multiplying them with
a test function: if $u\in \calD'(U)$ and
$f\in\calD(U)$, then $fu$ is a distribution with compact support in $U$ and
we can extend it to a distribution defined on $\bbR^n$ by setting it to equal to zero outside $U$.
Let us still denote by $fu$ this compactly supported distribution on
$\bbR^n$.
It has a Fourier transform $\widehat{fu}(k)$ which is an entire analytic
function of $k$ by the Paley--Wiener--Schwartz Theorem.

Following the physicist's
convention~\cite{Itzykson},\cite[p.~32]{Peskin}, we define the
Fourier transform of $u$ by
\begin{eqnarray*}
\calF(u)(k) &=& \hat{u}(k) = \int_{\bbR^n} dx e^{i k\cdot x} u(x),
\end{eqnarray*}
where $k\cdot x= \sum_i k_i x^i$ (we could interpret this quantity as an Euclidean scalar product
between two vectors in $\bbR^n$; however as we will see in Section \ref{oscillatory}
it is better to understand $k$ as a \emph{covector} and the product $k\cdot x$ as a
\emph{duality} product, this is the reason for the lower indices used for the coordinates
of $k$ and the upper indices used for the coordinates of $x$).
More rigorously, the above definition applies to functions
$f$ of rapid decrease and, for a tempered distribution
$u$, the Fourier transform is defined by
$\langle \hat{u},f\rangle=\langle u,\hat{f}\rangle$.
The inverse Fourier transform is
\begin{eqnarray*}
u(x) &=& \int \frac{dk}{(2\pi)^n} e^{-i k\cdot x} \hat{u}(k),
\end{eqnarray*}
where $n$ is the dimension of spacetime.
The same convention was used, for example,
by Franco and Acebal~\cite{Franco-07}.
Note the relation between this Fourier transform 
and the one used in other references:
$\hat{u}(k)=\mathcal{F}_H(u)(-k)$~\cite{HormanderI,Alinhac-07},
or $\hat{u}(k)=(2\pi)^{n/2}\mathcal{F}_{RS}(u)
(-k)$~\cite{ReedSimonII,Strohmaier,Wagschal-11}.

We can now give a definition of the product of two distributions.
Note that there are alternative definitions, under different
hypotheses (and we will meet another one later on).
For a general overwiew about the existing options, 
see \cite{Oberguggenberger1,Oberguggenberger}.
\begin{dfn}\label{definition-product-1}
Let $u$ and $v$ in $\calD'(\bbR^n)$. We say that
$w\in \calD'(\bbR^n)$ is the product of $u$ and $v$
if and only if, for each $x\in \bbR^n$, there exists some
$f\in \calD(\bbR^n)$, with $f=1$ near $x$, so that for
each $k\in \mathbb{R}^n$ the integral
\begin{eqnarray}
\widehat{f^2w}(k) &=& 
(\widehat{fu}\star \widehat{fv})(k)=
\frac{1}{(2\pi)^n} \int \widehat{fu}(q)
\widehat{fv}(k-q) \dd q,
\label{disprod}
\end{eqnarray}
is absolutely convergent.
\end{dfn}
When it exists, this product has many desirable properties: it is
unique, commutative, associative (when all intermediate products are
defined) and it coincides with the product 
of Theorem~\ref{singsupthm} when
the singular supports of $u$ and $v$ 
are disjoint~\cite[p.~90]{ReedSimonII}.

Let us consider some examples.
\begin{example}
If $u=v=\delta$, the product is not defined.
\end{example}
\begin{proof}
For any test function $f$ satisfying the hypothesis of the definition,
$f\delta(x)=f(0)\delta(x)=\delta(x)$ and
$\widehat{f\delta}(k)=1$, so that 
$\int \widehat{f\delta}(q) \widehat{f\delta}(k-q) \dd q
= \int \dd q $, which is not absolutely convergent.
\end{proof}

\begin{example}\label{exampleHeaviside}
 If $u = v = \theta$, the product is well defined.
\end{example}
\begin{proof}
 For any $f\in \mathcal{D}(\mathbb{R})$, $\widehat{f\theta}(k) =
\int_0^\infty e^{ikx}f(x)dx$ satisfies the
 uniform bounded $|\widehat{f\theta}(k)|\leq \|f\|_{L^1}:=
\int_\mathbb{R}|f(x)|dx$. Moreover an integration by
 part gives us also
 $\widehat{f\theta}(k) = \frac{i}{k}\left[f(0) + g(k)\right]$ with
$g(k):= \int_0^\infty e^{ikx}f'(x)dx$
 and we thus have the uniform bound $|\widehat{f\theta}(k)|\leq
\frac{1}{|k|}(|f(0)|+ \|f'\|_{L^1})$.
 Hence, for any $k\in \mathbb{R}$,
$\vert\widehat{f\theta}(k)\vert\leqslant C(1+\vert k\vert)^{-1} $
for $C=\|f'\|_{L^1}+\|f\|_{L^1}+|f(0)|$ and the
integral defining $(\widehat{fu}\star \widehat{fv})(k)$
is absolutely convergent because
\begin{eqnarray*}
  \int_\mathbb{R}\left|\widehat{f\theta}(q)\widehat{f\theta}(k-q)\right|dq
 \leq  \int_\mathbb{R}\frac{C^2 dq}{(|k-q|+1)(|q|+1)}
 \leq    \tilde{C} \int_\mathbb{R}\frac{dq}{(|q|+1)^2},
\end{eqnarray*}
where $\tilde{C}=C^2\sup_q\frac{(|q|+1)}{(|k-q|+1)}$
is finite.
\end{proof}

\begin{example}\label{example1surx}
\label{exuu}
If $u(x)=v(x)= 1/(x+i0^+)$, the product exists.
\end{example}
\begin{proof}
By contour integration, 
$\hat{u}(k)=-2i\pi\theta(-k)$. Thus,
\begin{eqnarray*}
\widehat{fu}(k) &=& \frac{1}{2\pi}
\int_{\mathbb{R}} d q \hat{f}(q) \hat{u}(k-q) = -i \int_{k}^\infty \hat{f}(q) d q,
\end{eqnarray*}
tends to
$-2\pi i f(0)=-2\pi i$ for $k\to -\infty$.

To show
that the integral in eq.~(\ref{disprod}) is absolutely
convergent, we define the smooth function
$F(k)=\int_k^{+\infty} \hat{f}(q) dq$.
The Fourier transform of a test function $f$ 
is fast decreasing: for any integer $N$, there is a constant
$C_N$ for which 
$|\hat{f}(q)|\le C_N (1+|q|)^{-N}$~\cite[p.~252]{HormanderI}.
Thus, 
for $k\ge 0$
\begin{eqnarray*}
|F(k)| &\le & C_N\int_{k}^\infty (1+q)^{-N}dq = \frac{
C_N}{N-1} (1+k)^{1-N},
\end{eqnarray*}
is fast decreasing and for any $k\in \bbR$
\begin{eqnarray*}
|F(k)| &\le & C_N\int_{-\infty}^\infty (1+|q|)^{-N}dq = \frac{2
C_N}{N-1}.
\end{eqnarray*}
Therefore, the right hand side of eq.~(\ref{disprod}) 
can be written 
$-(2\pi)^{-1}(\int_{-\infty}^k + \int_k^0+\int_0^{+\infty})F(q)
F(k-q) dq$.
The first integral is absolutely convergent because
$|F(q)F(k-q)|\le 2 C_N^2 (N-1)^{-2} (1+|k-q|)^{1-N}$,
the second because
the integrand is smooth and the domain is finite
and the third integral because
$|F(q)F(k-q)|\le 2 C_N^2 (N-1)^{-2} (1+q)^{1-N}$.

To compute the product $w=u^2$ we take $f=1$
and we calculate directly
\begin{eqnarray*}
\widehat{u^2}(k) &=& \frac{1}{2\pi} \int_{\bbR} \hat{u}(q)
\hat{u}(k-q) \dd q = -2\pi \int_{\bbR}\theta(-q)\theta(q-k)
= 2\pi k \theta(-k).
\end{eqnarray*}
\end{proof}
Note that
the Fourier transform of the derivative of a distribution $v$ is
given by $\widehat{v'}(k)=-ik \hat{v}(k)$.
Thus we recover the relation $\widehat{u^2}=-\widehat{u'}$
i.e. $u(x)^2 = (x+i0^+)^{-2} = -\frac{d}{dx}(x+i0^+)^{-1}$.

\begin{example}
\label{exuv}
If $u(x)=1/(x+i0^+)$ and
$v(x)=1/(x-i0^+)$, the product does not
exist.
\end{example}
\begin{proof}
We have $\hat{u}(k)=-2i\pi\theta(-k)$ and
$\hat{v}(k)=2i\pi\theta(k)$. Thus,
\begin{eqnarray*}
\widehat{fv}(k) &=& \frac{1}{2\pi}
\int d q \hat{f}(q) \hat{v}(k-q) = i \int_{-\infty}^k \hat{f}(q) d q,
\end{eqnarray*}
which decreases fast for $k\to -\infty$ and tends to
$2\pi i f(0)=2\pi i$ for $k\to +\infty$.
We define $G(k)=\int_{-\infty}^k \hat{f}(q) dq$
and recall that
 $F(k)=\int_k^{+\infty}\widehat{f}(q)dq$
so
that $F+G=2\pi$.
The right hand side of eq.~(\ref{disprod}) 
can be written as the limit for $M\to \infty$ of 
$(2\pi)^{-1} I_M(k)$ with
\begin{eqnarray*}
I_M(k) &=& \int_{-M}^\infty 
   F(q) G(k-q) dq \\&=& 2\pi \int_{-M}^\infty  F(q) dq
   - \int_{-M}^\infty F(q) F(k-q) dq .
\end{eqnarray*}
We saw in the previous example that the second term
is absolutely convergent and for the first term we 
use $F=2\pi-G$ to write
\begin{eqnarray*}
\int_{-M}^\infty  F(q) dq &=&
  \int_0^\infty F(q) dq + 2\pi\int_{-M}^0 dq
  - \int_{-M}^0 G(q) dq.
\end{eqnarray*}
The decay properties of $F$ and $G$ imply that
the first and third terms are absolutely convergent,
but the second term is $2\pi M$ which diverges
for $M\to\infty$. Thus, there is no test function
$f$ with  $f(0)=1$ such that $I_M(k)$ converges:
the product of distributions does not exist.
\end{proof}

In example~\ref{exuu}, the distribution $u^2$
was calculated without using the localizing
test function $f$. In general this is not possible.
For example, consider
\begin{example}\label{exampleu1u2}
\label{exupv}
\begin{eqnarray*}
u(x) &=& 
\frac{1}{x+i0^+} 
+
\frac{1}{x+a-i0^+},
\end{eqnarray*}
with $a\not=0$.
Then, $u^2$ exists.
\end{example}
Indeed, denote by $u_1$ and $u_2$ the
two terms on the right hand side.
We showed that $u_1^2$ exists and
the same reasoning implies that 
$u_2^2$ exists. The cross term
$u_1u_2$ exists because the singular support
of $u_1$, which is $\{0\}$, is disjoint from 
the singular support
of $u_2$, which is $\{-a\}$.
Thus, $u^2$ exists although the Fourier transform of
$u$ (i.e.  $\hat{u}(k)=-2i\pi\theta(-k)
+2i\pi e^{-ika}\theta(k)$)
is slowly decreasing in both
directions.
Therefore, the role of the localizing test function $f$
is not only to make the Fourier transform of $fu$
exist (even when the Fourier transform of $u$ does not),
but also to isolate the singularities of $u$.
In example \ref{exupv},
the two singular points of $u$ are $x=0$ and $x=-a$.
To localize the distribution around $x=0$, 
we multiply $u$ by a smooth
function $f$ such that $f(0)=1$ and
$f(x)=0$ for $|x|> |a|/2$, so that
$\widehat{fu}(k)=-i\int_k^\infty \hat{f}(q)dq$ is 
fast decreasing in the direction of $k>0$
because the contribution of $1/(x+a-i0^+)$ 
is eliminated. Conversely, if we multiply the distribution
by a smooth function $g$ such that $g(-a)=1$
and $g(x)=0$ for $|x+a|>|a|/2$, then 
$\widehat{gu}(k)=i\int_{-\infty}^k \hat{g}(q)dq$,
which is fast decreasing in the direction $k<0$.\\

\subsubsection{Discussion}
In the previous examples, we saw that the calculation of 
the product of two distributions by using the Fourier transform
looks rather tricky. In particular, it seems that we
have to know the Fourier transform of the 
product of each distribution with an arbitrary function.

Moreover even when we are able to define it,
the product of distribution does not always
satisfy the Leibniz rule $\partial (uv) = (\partial u)v + u(\partial v)$.
For instance the product of $\theta$ makes sense (Example~\ref{exampleHeaviside})
but does not respect the Leibniz rule (see Section \ref{multdissect}).
On the other hand the square of $1/(x+i0^+)$ can be defined (see Example \ref{example1surx})
and this definition agrees with Leibniz rule.

Fortunately, H\"ormander devised a powerful condition on a pair of distributions to:
1) guarantee the existence of their product
without computing it; 2) ensure that this product satisfies the
Leibniz rule.

As a preparation for this condition, we can
analyze why the product exists in 
example~\ref{exuu} and not in example~\ref{exuv}.
In example~\ref{exuu}, the support of $\hat{u}$  is
$({-}\infty,0)$ and, because of the convolution
formula $\hat{u}(q) \hat{u}(k-q)$, the support of
$\hat{u}(q) \hat{u}(k-q)$ as a function of $q$
is the finite interval ${[}k,0{]}$ if $k\le0$ and
is empty if $k>0$. Thus, the integral over $q$ is 
absolutely convergent.
On the other hand, in example~\ref{exuv}
the support of 
$\hat{u}(q) \hat{v}(k-q)$ is 
$({-}\infty,\min(k,0))$, which is infinite.

In general, for the convolution integral to be well defined, 
we just need that the product
$\widehat{fu}(q) \widehat{fv}(k-q)$ 
decreases fast enough for large $q$ for the integral
over $q$ to be absolutely convergent.
Note also that, for any distribution $u$ and for any smooth function $f$
with compact support, since $fu$ is a distribution with compact support,
its Fourier transform
$\widehat{fu}$ grows at most polynomially at infinity, i.e. there exists
some $p\in \mathbb{N}$
and some constant $C>0$ such that $|\widehat{fu}(k)|\leq C(1+|k|)^p$
everywhere.
Hence it is enough that one of the two factors in the product
$\widehat{fu}(q) \widehat{fv}(k-q)$ is fast decreasing
at infinity to ensure that the product is fast decreasing.
In example~\ref{exuu}, $\widehat{fu}(q)$
decreases very fast for 
$q\to+\infty$ but does not decrease
for $q\to-\infty$. 
If $\widehat{fu}(q)$ decreases slowly in
some directions $q$, this must be
compensated by a fast decrease of
$\widehat{fv}(k-q)$ in the same direction $q$.
This is exactly what happens in
example~\ref{exuu} and not in 
example~\ref{exuv}.

Lastly Example \ref{exampleu1u2} confirms that a general condition for the existence
of a product
of distributions should use the Fourier transform
of distributions localized around singular points.

It is now time to introduce the key notion for defining 
H{\"o}rmander's product of distributions:
the \emph{wavefront set}.

\section{The wavefront set}
\label{guessWFsect}
We want to find a sufficient condition by which the product of
distributions defined in eq.~\eqref{disprod} is 
absolutely convergent. In this integral,
the distribution $fv$ is compactly supported because
$f\in\calD(\bbR^n)$. Thus, there is constant $C$
and an integer $m$ such that
$|\widehat{fv}(k-q)|\le C(1+|k-q|)^m$. The smallest $m$ for which
this inequality is satisfied is called the \emph{order}
of the distribution $fv$.
The integral \eqref{disprod} would be absolutely convergent if we had
$|\widehat{fu}(q)|\le C'(1+|q|)^{-m-n-1}$.
However, since we also wish the product of distributions
to be compatible with
derivatives through the Leibniz rule
$\partial(uv)=(\partial u)v + u(\partial v)$ and since 
a derivative of order $n$ increases the order of $u$ by $n$,
what we really need is that 
$|\widehat{fu}(q)|$ decreases faster than any
inverse power of $1+|q|$.
We give now a precise definition of 
the property of fast decrease.

\subsection{Outside the wavefront set: the regular directed points}
We start by defining some basic tools: the conical
neighborhoods and the fast decreasing functions.
\begin{dfn}
A \emph{conical neighborhood} of a point $k\in \bbR^n\setminus \{0\}$
is a set $V\subset \bbR^n$ such that $V$ contains 
the ball $B(k,\epsilon)=\{q\in \bbR^n\telque |q-k| < \epsilon\}$
for some $\epsilon >0$ and, for any $p$ in $V$ and any $\alpha >0$,
$\alpha p$ belongs to $V$.
\end{dfn}
An example of conical neighborhood of $k$ is given in figure~\ref{figconical}.
\begin{figure}
\begin{center}
\includegraphics[width=8.0cm]{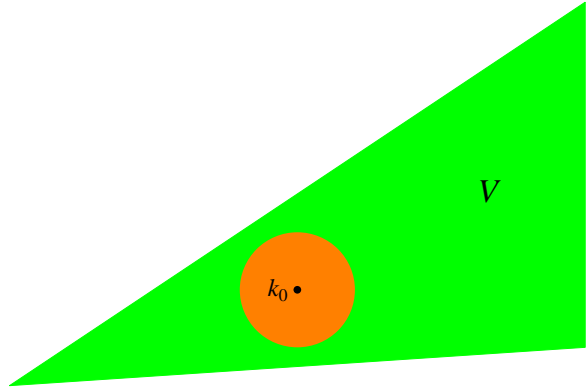}
\caption{Example of a conical neighborhood of
   $k$. 
  \label{figconical}}
\end{center}
\end{figure}
\begin{dfn}
A smooth function $g$ is said to be \emph{fast decreasing} on
a conical neighborhood $V$ if, for any integer $N$,  there is a constant 
$C_N$ such that
$|g(q)|\le C_N (1+|q|)^{-N}$ for all $q\in V$.
\end{dfn}
For example, the function $e^{-q^2}$ is fast decreasing
on $\bbR^n$. We need functions to be fast
decreasing in a conical neighborhood and not only along
a specific direction (which would be the case if 
$C_N$ were a function of $q$), because a single
direction has zero measure and we would not be
able to control the integral \eqref{disprod}.
According to the discussion of the previous
section we see that the integral \eqref{disprod}
converges absolutely if the directions where
$\widehat{fv}(k-q)$ decrease slowly correspond
to regions where $\widehat{fu}(q)$ is fast decreasing.

We define now the ``nice points'' around which $\widehat{fu}$
is fast decreasing. They are called
\emph{regular directed points}~\cite[p.~92]{ReedSimonII}:
\begin{dfn}
For a distribution $u\in \calD'(\bbR^n)$, 
a point $(x,k)\in \bbR^n\times (\mathbb{R}^n\backslash\{0\})$
is called a \emph{regular directed point} of $u$ if and only if
there exist: (i) a function $f\in \calD(\bbR^n)$ with
$f(x)=1$ and (ii) a closed conical neighborhood  $V\subset \bbR^n$ 
of $k$, such that 
$\widehat{fu}$ is fast decreasing on $V$.
\end{dfn}
The relevance of the concept
of regular directed point also stems from the 
following theorem~\cite[p.~252]{HormanderI}
\begin{thm}
\label{fastthm}
A compactly supported distribution
$u\in \mathcal{E}'(U)$ is a smooth function if and only if 
$\widehat{u}(q)$ is fast decreasing on $\bbR^n$.
\end{thm}
This theorem is physically reasonable because,
if $f$ is a smooth function, then $f(x)e^{ik\cdot x}$
oscillates widely when $k$ is large, so that the
average of this expression (i.e. $\widehat{f}(k)$) is
very small.
Theorem~\ref{fastthm} implies that any singularity
of a distribution can be detected by an absence of
fast decrease in some direction: a point
$x$ is in the singular support if and only if there is a direction
$k$ where the Fourier transform is not fast decreasing.
However, if $x\in \singsupp u$, there can be
directions $k$ such that $(x,k)$ is regular directed.
In example~\ref{exuu}, we saw
that $\widehat{fu}(k)$ is rapidly decreasing for $k>0$ but
not for $k<0$. This brings us finally to the definition of the 
wavefront set

\subsection{The definition of the wavefront set and the Product Theorem}
\begin{dfn}
\label{defWFRS}
The \emph{wavefront set} of a distribution $u\in \calD'(\bbR^n)$
is the set, denoted by $\WF(u)$,
of points $(x,k)\in \bbR^n\times (\mathbb{R}^n\backslash\{0\})$
which are not regular directed for $u$.
\end{dfn}

In other words, for each point
of the singular support of $u$, the wavefront set of $u$ is composed of the directions
where the Fourier transform of $fu$ is not fast decreasing, for $f$
a sufficiently small support.
The name ``wavefront set'' comes from the fact that
the singularities of the solutions of the wave equation
move within it~\cite[p.~274]{HormanderI}, so that
the wavefront set describes the evolution of the wavefront.
The wavefront set is a refinement of the singular support,
in the sense that the singular support of $u$
is the set of points $x\in \bbR^n$, such that
$(x,k)\in \WF(u)$ for some nonzero $k\in \bbR^n$.

Now we see how this definition can be used to
determine the product of two distributions $u$
and $v$. Broadly speaking,
if a point $x$ belongs to the singular
support of $u$ and $v$, then the product of 
$u$ and $v$ exists at $x$ if, for all
directions $q$, either 
$\widehat{fu}(q)$ or 
$\widehat{fv}(k-q)$ is rapidly decreasing.
In particular, if $(x,q)$ belongs to $\WF(u)$,
then $(x,-q)$ must not belong to $\WF(v)$.
This is called \emph{H\"ormander's condition}
and the precise theorem is~\cite[p.~267]{HormanderI}:
\begin{thm}[Product Theorem]
\label{disprodthm}
Let $u$ and $v$ be distributions in $\calD'(U)$. Assume that
there is no point $(x,k)$ in $\WF(u)$ such that
$(x,-k)$ belongs to $\WF(v)$,
then the product $uv$ can be defined.
Moreover, if so, then
\begin{equation}\label{conclusion-mthm}
 \WF(uv)\subset S_+ \cup S_u \cup S_v,
\end{equation}
where $S_+=\{(x,k+q)| 
(x;k)\in \WF(u)\,\, \mathrm{and}\,\, 
(x;q)\in \WF(v)\}$,
$S_u=\{(x;k)| (x;k)\in \WF(u)\,\,\mathrm{and}\,\, x\in \supp(v)\}$ and
$S_v=\{(x;k)| (x;k)\in \WF(v)\,\,\mathrm{and}\,\, x\in \supp(u)\}$.
\end{thm}
\noindent
\textbf{Remarks}
\begin{enumerate}
\item This theorem is absolutely fundamental for the
theory of renormalization in curved spacetimes. With this simple
criterion, we can prove that a product of distributions
exists even if we cannot calculate their Fourier transforms
and even if we do not know the explicit form of
the distributions.
 \item The condition involving the support of $u$  
in $S_v$ and the support of $v$ in $S_u$
in $\WF(uv)$
is given in \cite[p.~84]{Grigis} but
is usually not stated explicitly~\cite[p.~267]{HormanderI}
\cite[p.~21]{Duistermaat}
\cite[p.~95]{ReedSimonII}, \cite[p.~527]{Chazarain},
\cite[p.~153]{Friedlander}, \cite[p.~193]{Strichartz-03},
\cite[p.~61]{Eskin}. This support condition is
imperative to calculate the
wavefront set of example~\ref{delta1delta2}
or of the Feynman propagator in section~\ref{Feynpropsect}.
\item When H\"ormander's condition holds,
then the product of distributions satisfies the Leibniz
rule for derivatives, because derivatives
do not extend the wavefront set~\cite[p.~256]{HormanderI}).
\item Note that if $u$ and $v$ satisfy H\"ormander's condition, then 
their product exists in the sense of Definition~\ref{definition-product-1}. 
The converse is not true in general. However,
if the product of distributions is extended beyond H\"ormander's
condition, then it is generally not compatible with
the Leibniz rule, as shown by the example of the
Heaviside distribution at the beginning 
of section~\ref{multdissect}.
\item H\"ormander's condition of the Product Theorem can be rephrased
by saying that $S_+$ does not meet the zero section 
(of the cotangent bundle over $U$), i.e. that
$S_+\cap (U\times \{0\}) = \emptyset$.
\item For any pair $A$ and $B$ of subsets of $U\times \bbR^n$, we can define
$A\caplus  B:= \{(x,k+q)|(x,k)\in A,(x,q)\in B\}$. We then observe that
$S_+ = WF(u)\caplus WF(v)$ and hence H\"ormander's condition
amounts to saying that
$WF(u)\caplus WF(v)$ does not intersect the zero section.
On the other hand if we set
$\underline{WF}(u):= WF(u)\cup(\hbox{supp }u\times \{0\})$, etc.,
we then always have $\underline{WF}(u)\caplus  \underline{WF}(v) =
S_+\cup S_u\cup S_v\cup (\hbox{supp }(uv)\times \{0\})$. Moreover if 
H\"ormander's condition holds then
$\hbox{supp }(uv)\times \{0\}$ is disjoint from $S_+\cap S_u\cap S_v$
and thus Conclusion (\ref{conclusion-mthm}) is equivalent to the inclusion
$\underline{WF}(uv)\subset \underline{WF}(u)\caplus \underline{WF}(v)$.
\end{enumerate}

\subsection{Simple examples and applications of the Product Theorem}
We give a few very simple examples.

\begin{example}
The simplest example is $\delta(x)$ in $\calD'(\mathbb{R}^n)$,
for which $\WF(\delta)=\{(0;k)| k\in\mathbb{R}^n, k\not=0\}$. Thus,
the powers of $\delta$ cannot be defined.
\end{example}
\begin{proof}
The singular support of $\delta(x)$ is $\{0\}$
and $\widehat{f \delta}(k)=f(0)$ is not fast decreasing if
$f(0)\not=0$. This proves that 
$\WF(\delta)=\{(0;k)| k\in\mathbb{R}^n, k\not=0\}$. 
To show that the product is not allowed,
consider any point $(0;k)$ of $WF(\delta)$, then
$(0;-k)$ is also a point of $\WF(\delta)$ and the 
H\"ormander condition is not satisfied. 
\end{proof}
\begin{example}
\label{Heavisideex}
The wavefront set of the Heaviside distribution
$\theta$ is $\WF(\theta)=\{(0;k)\telque k\not=0\}$.
There is a constant $C$ such that
$|\widehat{\theta f}(k)|\le C/(1+|k|)$ for all $k$.
\end{example}
\begin{proof}
The Heaviside distribution is smooth for $x<0$ and
$x>0$ because it is constant there. Thus, the only
possible singular point is $x=0$.
Consider a smooth compactly supported function $f$
such that $f(0)=1$. We have for $k\not=0$
\begin{eqnarray}
\widehat{\theta f}(k) &=& \int_0^\infty e^{ikx} f(x) dx 
 = \frac{(-i)}{k} 
       \int_0^\infty (e^{ikx})' f(x) dx 
\nonumber\\&=&
 \frac{i f(0) }{k} + \frac{i}{k} 
              \int_0^\infty e^{ikx} f'(x) dx,
\label{identity1-on-theta}
\end{eqnarray}
where the prime denotes a derivative with respect to $x$
and we integrated by parts. A further integration by part
gives us
\begin{equation}\label{identity2-on-theta}
\widehat{\theta f}(k) = \frac{i f(0) }{k} 
- \frac{f'(0)}{k^2} - \frac{1}{k^2} 
              \int_0^\infty e^{ikx} f''(x) dx.
\end{equation}
Let $L$ be the length of $\supp f$ and, for $n=0,1,2$, let $M_n$ be a constant 
such that $|f^{(n)}(x) | \le M_n$ for all $x$.
Using $f(0)= 1$,
Identity (\ref{identity2-on-theta}) implies that $|\widehat{\theta f}(k) - \frac{i}{k}|\leq
\frac{M_1+LM_2}{k^2}$.  Hence $(0,k)\in WF(\theta)$, $\forall k\neq 0$.
On the other hand (\ref{identity1-on-theta}) implies both $|\widehat{\theta f}(k)|\leq LM_0$
and $|\widehat{\theta f}(k)|\leq \frac{1+LM_1}{|k|}$. We hence deduce that
$|\widehat{\theta f}(k)|\leq \frac{C}{1+|k|}$, for some constant $C$.
\end{proof}
The wavefront set of $\theta$ is the same as the wavefront set of
$\delta$. This explains why the powers of $\theta$ are
not allowed in the sense of H\"ormander. 
\begin{example}
$u(x)=1/(x+i0^+)$,
then $\WF(u)=\{(0;k), k<0\}$. 
Thus, $u^2$ exists and $\WF(u^2)=\WF(u)$.
\begin{proof}
The proof is obvious from example~\ref{exuu}
(see also~\cite[p.~94]{ReedSimonII}, where
one must recall that the sign is opposite because
of the different convention for the Fourier transform).
\end{proof}
\end{example}
\begin{example}
$v(x)=1/(x-i0^+)$,
then $\WF(v)=\{(0;k), k>0\}$. 
Thus, $v^2$ exists and $\WF(v^2)=\WF(v)$,
but we cannot conclude that
$uv$ exists (it does not, as we saw in example~\ref{exuv}).
\end{example}
\begin{example}
We consider again example~\ref{exupv}:
\begin{eqnarray*}
u(x) &=& 
\frac{1}{x+i0^+} 
+
\frac{1}{x+a-i0^+},
\end{eqnarray*}
with $a\not=0$.
Then, 
$\WF(u)=\{(0;k), k<0\} \cup \{(-a;k), k>0\}$
and $u^2$ exists, with $\WF(u^2)=\WF(u)$.
\end{example}
\begin{example}
\label{delta1delta2}
(See \cite[p.~97]{ReedSimonII}).
Let $\delta_1$ and $\delta_2$ be the distributions in 
$\calD'(\bbR^2)$ defined by $\langle \delta_1,f\rangle=\int d y f(0,y)$
and $\langle \delta_2,f\rangle=\int d x f(x,0)$. Then,
$\WF(\delta_1)=\{(0,y;\lambda,0) | 
y\in \mathbb{R}, \lambda\not=0\}$ and
$\WF(\delta_2)=\{(x,0;0,\mu) | 
x\in \mathbb{R}, \mu\not=0\}$.
Thus, $\delta_1\delta_2$ exists and
$\WF(\delta_1\delta_2) \subset
\{(0,0;\lambda,\mu), \lambda\not=0,\mu\not=0\}
\cup \{(0,0;\lambda,0), \lambda\not=0\} 
\cup \{(0,0;0,\mu), \mu\not=0\} $, where
we used 
$\supp(\delta_2)=\{(x,0)| x\in \bbR\}$
and
$\supp(\delta_1)=\{(0,y)| y\in \bbR\}$.
Note that the estimate of the wavefront set of
$\delta_1\delta_2$ would be much worse if
the support of $\delta_2$ and $\delta_1$
had not been taken into account in 
$S_{\delta_1}$ and $S_{\delta_2}$ of the
Product Theorem.
In that case the inclusion is in fact
an equality because
$\WF(\delta_1\delta_2) =
\{(0,0;\lambda,\mu), 
(\lambda,\mu)\not=(0,0)\}$.
\end{example}
\begin{proof}
Let $y\in\mathbb{R}$, we want to calculate $WF\left(\delta_1\right)$
at $(0,y)$.  Take a test function 
$f(x_1,x_2)$ which is equal to one
around $(0,y)$. Then,
\begin{eqnarray*}
\widehat{f\delta_1}(k) &=&
\int d x_1 d x_2 f(x_1,x_2) \delta(x_1) 
e^{i k_1 x_1+i k_2 x_2} =
\int d x_2 f(0,x_2) 
e^{i k_2 x_2}.
\end{eqnarray*}
Take $k=(k_1,k_2)$ and observe the decay of 
$\widehat{f\delta_1}(\lambda k)$. If 
$k_2\not=0$ this is a fast decreasing function of 
$\lambda$ because $f(0,x_2)$ is a smooth compactly
supported function of $x_2$. 
If $k_2=0$, then we have
$\widehat{f\delta_1}(k_1,0) =
\int d x_2 f(0,x_2)$,
which is independent of $k_1$, so that 
$\widehat{f\delta_1}(\lambda k_1,0)$ is not fast
decreasing. This proves that
$\WF(\delta_1)$ has the given form.
A similar proof yields $\WF(\delta_2)$.
The rest follows from the fact that
$\delta_1\delta_2$ is the two-dimensional delta function.
\end{proof}

\section{The wavefront set of a characteristic function}
Now that we know the definition of the wavefront set,
we shall get the feel of it by studying in detail
the characteristic distribution $u$ of a region $\Omega$ of 
$\bbR^n$, defined by
$\langle u,f\rangle = \int_\Omega f(x) dx$. We shall also revisit
it in section~\ref{presentation-Radon}.

\subsection{The upper half-plane}\label{half-plane}
For concreteness we start from
the characteristic distribution of the upper half-plane
\begin{eqnarray*}
\langle u,f\rangle &=& \int_{-\infty}^\infty dx_1 \int_0^\infty dx_2
f(x_1,x_2).
\end{eqnarray*}
This is the distribution corresponding to the function
equal to one on the upper half-plane
(i.e. if $x_2\ge0$) and to zero on the lower half-plane
(i.e. if $x_2<0$).
It is intuitively clear that the singular support
of $u$ is the line $(x_1,0)$. Now take a point
$(x_1,0)$ of the singular support and a test
function $f$ which is non-zero on $(x_1,0)$. What are the
directions of slow decrease of $\widehat{fu}$?
It seems clear that $\widehat{fu}(k)$ decreases fast when
$k$ is along $(1,0)$, because we do not feel the step
of $u$ if we walk along it and do not cross it. 
But what about the other directions? Does the wavefront
set contain all the directions that cross the step
or just the direction $(0,1)$ which is perpendicular to it?

The wavefront set of $u$ can be obtained by noticing that
$u$ is the (tensor) product of the constant function 1
for the variable $x_1$ by the Heaviside distribution
$\theta(x_2)$. Then, a standard
theorem~\cite[p.~267]{HormanderI} gives us
$\WF(u)=\{(x_1,0;0,\lambda),\lambda\not=0\}$. 
In other words, the wavefront set 
detects the direction perpendicular to the
step. It is instructive to make an explicit
calculation to understand why it is so.

We use an idea of Strichartz~\cite[p.~194]{Strichartz-03}
and consider test functions $f(x_1,x_2)=f_1(x_1)f_2(x_2)$.
This is not really a limitation because any test
function can be approximated by a finite sum of such products.
Then $\widehat{uf}(k)=\widehat{f_1}(k_1)\widehat{\theta f_2}(k_2)$.
We want to show that, if $k_1\not=0$, for every integer $N$
there is a constant $C_N$ such that
$|\widehat{uf}(\tau k)|\le C_N (1+\tau|k|)^{-N}$
for every $\kappa >0$.
We already know that there is a constant $D_N$ such that
$|\widehat{f_1}(\tau k_1)|\le D_N (1+\tau|k_1|)^{-N}$ because
$f_1$ is smooth and a constant $C$ such that
$|\widehat{f_2}(\tau k_2)|\le C (1+\tau|k_2|)^{-1}$
(see Example~\ref{Heavisideex}).
We are going to show that, if the component $k_1$ of 
$k$ is not zero, the fast decrease of $\widehat{f_1}(\tau k_1)$
induces the fast decrease of $\widehat{uf}(\tau k)$.
If $k_1\not=0$, we have $|k|\le \alpha |k_1|$ 
where $\alpha=|k|/|k_1|$.
Note that $\alpha\ge 1$ because $|k_1|\le |k|$.
Hence $(1+\tau|k|) \le \alpha (1+\tau|k_1|)$
and 
\begin{eqnarray*}
|\widehat{uf}(\tau k)|\le  C D_N 
    (1+\tau|k_1|)^{-N} (1+\tau|k_2|)^{-1}
  \le C D_N\alpha^N (1+\tau|k|)^{-N},
\end{eqnarray*}
where we bounded $(1+\tau|k_2|)^{-1}$ by $1$.
Finally, if $k_1\not=0$, then
$|\widehat{uf}(\tau k)|\le  C_N(1+\tau|k|)^{-N}$
for all $\kappa>0$ with $C_N=\alpha^N C D_N$. 
This result was obtained
for a single vector $k$, but it can be extended to
a cone around $k$ by increasing the value of $\alpha$.

\subsection{Characteristic function of general domains}
More generally, we can consider the characteristic function
of any domain $\Omega$ in $\bbR^n$ limited by a smooth surface $S$.
The characteristic function of $\Omega$ is the function
$\chi_\Omega$ such that $\chi_\Omega(x)=1$ if $x\in \Omega$
and $\chi_\Omega(x)=0$ if $x\notin \Omega$. The characteristic
function $\chi_\Omega$ corresponds to a distribution $u_\Omega$
defined by $\langle u_\Omega,f\rangle=\int_\Omega f(x) dx$.
The wavefront set of $u_\Omega$ is given by~\cite[p.~129]{Taylor-81}:
\begin{prop}
\label{exdisk}
Let $\Omega\subset \mathbb{R}^n$ be a region with smooth
boundary $S$ and let $u_\Omega=\chi_\Omega$ be the characteristic
distribution of $\Omega$.
Then $\WF (u_\Omega)= \{(x,k); x\in S, 
\,\,\mathrm{and}\,\,k\,\,\mathrm{normal}\,\,\mathrm{to}
\,\,S\}$.
\end{prop}
\begin{figure}
\begin{center}
\includegraphics[width=7.0cm]{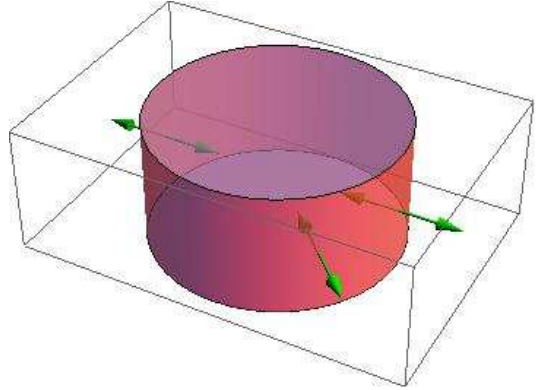}
\caption{The characteristic function of the unit disk
(pink) is equal to 1 for $x^2+y^2\le 1$ to zero for
$x^2+y^2 > 1$. 
Some vectors of the wavefront set are
indicated as green arrows.
For a given point $(x,y)$ of the boundary
$x^2+y^2=1$, the points $(x,y;k_x,k_y)$ of the wavefront
set are such that $(k_x,k_y)$ is perpendicular to the
boundary, thus $(k_x,k_y)=(\lambda x, \lambda y)$ for 
all $\lambda\not=0$. 
In this figure we represent the characteristic function,
the tangent bundle and the cotangent bundle in the
same coordinates.
  \label{figdisk}}
\end{center}
\end{figure}
Notice that the vectors $k$ are perpendicular to the boundary
$S$ of $\Omega$ (see fig.~\ref{figdisk} for the 
example of a disk). This can be understood by a hand-waving
argument. Since the boundary $S$ is smooth, by using
a test function with very small support around $x\in S$,
the boundary looks flat around $x$ and we can apply the argument of
the upper-half plane (generalized to $\bbR^n$) previously discussed.
The set of vectors $k$ 
which are perpendicular to all tangent vectors
to $S$ at $x$ is called the
\emph{conormal} of $S$ at $x$ and is denoted
by $C_x$ (see fig.~\ref{figdisk} for the example
where $n=2$ and $\Omega$ is the unit disk).
The set $C=\{(x,k)\telque k\in C_x\}$
is called the \emph{conormal bundle} of $S$.
The previous proposition says that the wavefront set
of $u_\Omega$ is the conormal bundle of $S$.

The wavefront set of a characteristic distribution
has many applications. Its ability to 
give an accurate description of the boundary of shapes
makes it particularly efficient
for image analysis~\cite{Candes-05}
and tomography~\cite{Xu-09}.

\subsection{Counting intersections}
We close this section by showing that the wavefront set of the characteristic
distribution of a bounded smooth domain $\Omega$ in the plane can be determined by
the following striking procedure. For each straight line $L_{k,a}$ in the plane,
denote by $n_{k,a}$ the number of times the straight line intersects the
boundary. 
\begin{figure}
\begin{center}
\includegraphics[width=7.0cm]{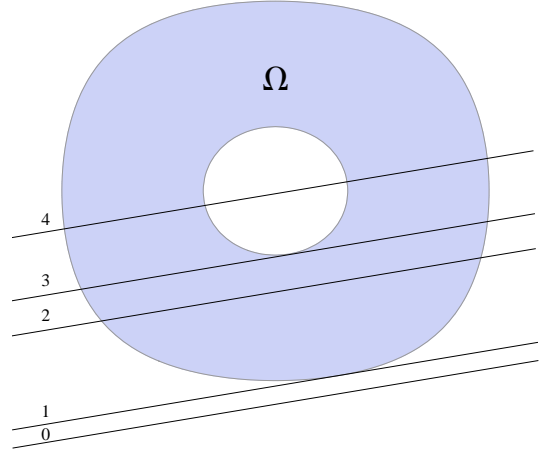}
\caption{Counting the number of times a straight line
crosses the boundary of $\Omega$.
  From the bottom
  to the top, this number is 0, 1, 2, 3 and 4.
It is possible to reconstruct $\Omega$ from the set of straight lines
and their numbers of crossings.
  \label{toporadonfig}}
\end{center}
\end{figure}
For generic domains, the wavefront set of $u_\Omega$ can be
recovered from the set of integers $n_{k,a}$~\cite{Viet-Radon}.
In particular, this information is sufficient to recover the shape of $\Omega$.
This remark can have applications in  image analysis.

In some exceptional cases, this result holds only up to localization or
the replacement of the number of intersections by the number of 
connected parts of the intersection~\cite{Viet-Radon}.
This characterization of the wavefront set can
be extended to surfaces in $\mathbb{R}^3$ 
if we replace the number of intersections $n_{k,a}$
by the Euler characteristic of the intersection of a given surface
with all possible half--spaces~\cite{Viet-Radon}.

\section{Use of the Radon transform}

\subsection{The wavefront set of a measure supported by a hypersurface}
In an attempt to better understand the wavefront set, we came up
with the following idea. 
As seen in Example c) in Section \ref{inwhichcases}, 
a distribution may be singular and may enjoy partial regularity properties
simultaneously. Consider for instance a smooth
submanifold $\Gamma\subset \mathbb{R}^n$ and the distribution which is the measure
$\mu$ supported by $\Gamma$ with the Euclidean density.
The singular character of $\mu$
shows up by restricting $\mu$ to a smooth path which crosses transversally
$\Gamma$: this gives us a Dirac mass type singularity. However if we probe $\mu$ by
moving in a parallel to $\Gamma$ we may be tempted to say that heuristically the distribution varies
smoothly.
Such a test cannot be performed by following a path which lies inside $\Gamma$,
because the restriction of $\mu$ to such a path would not make sense!
However we may replace such a path by a dual wave. In the most naive approach, this consists in
a family of hypersurfaces $(H_t)_t$ which cross transversally 
(e.g. orthogonally) our path
and which forms locally a foliation of an open subset
of $\mathbb{R}^n$. Each $H_t$ can be thought as a wavefront in this Huygens type picture.
This is another indication that we must interpret $p$ as a covector.

Let's explore this idea in the simple case where $\Gamma$ is a smooth curve. Choose a point $x_0\in \Gamma$
and a covector $p\in \mathbb{R}^n$, and define the linear form $\alpha:\mathbb{R}^n\longrightarrow \mathbb{R}$
by $\alpha(x):= p\cdot x$ and assume that
$\alpha|_{T_{x_0}\Gamma}\neq 0$.
We will test $\mu$ locally around $x_0$ by using a plane wave whose wavefronts
are the hyperplanes $H_{\alpha,a}$ of equation $\alpha(x)=a$, for
$a\in \mathbb{R}$ close to $\alpha(x_0)$. 
Choose an open neighborhood $U\subset\mathbb{R}^n$ of $x_0$ such that
there exists a parametrization $\gamma:I\longrightarrow U$
of $\Gamma\cap U$. Then for any $\varphi\in \mathcal{D}(\mathbb{R}^n)$
with support contained in $U$, we have
\[
 \langle \mu,\varphi\rangle = \int_I\varphi(\gamma(t))|\dot{\gamma}(t)|dt.
\]
Moreover we may choose $U$ such that $\alpha|_{T_x\Gamma}\neq 0$, $\forall x\in \Gamma\cap U$.
We remark then that  $\alpha\circ\gamma$ is a diffeomorphism into its image.

Let $\omega$ be an open subset of $U$ such that $\overline{\omega}\subset U$ and let $\chi\in \mathcal{D}(\mathbb{R}^n)$
with support contained in $U$ and such that $\chi=1$ on $\overline{\omega}$. Let 
$f\in \mathcal{D}(\mathbb{R})$ with support in $\alpha(\omega\cap \Gamma)$. Set $\varphi:= \chi(f\circ \alpha)$
and observe that $f\circ \alpha = \varphi$ on $U\cap \Gamma$.
Hence we can define $\langle \mu,f\circ \alpha\rangle$ by setting
\[
 \langle \mu,f\circ \alpha\rangle:= \langle \mu,\varphi\rangle
 = \int_If\circ \alpha\circ\gamma(t)|\dot{\gamma}(t)|dt.
\]
By performing the change of variable $a=\alpha\circ\gamma(t)$, $da = |\alpha(\dot{\gamma}(t))|dt$,
$A = \alpha\circ \gamma(I)$, we obtain
\[
 \langle \mu,f\circ \alpha\rangle = \int_A f(a)\frac{da}{|\alpha(\tau(a))|},
\]
where $\tau(a)$ is the tangent vector to $\Gamma$: $\tau(\alpha\circ \gamma(t)) = \dot{\gamma}(t)/|\dot{\gamma}(t)|$.
We see that we can extend this definition by replacing $f$ by a Dirac mass $\delta_a$ at some value
$a\in A$. We then get $\langle \mu,\alpha^*\delta_a\rangle = 1/|\alpha(\tau(a))| = 
1/| p\cdot \tau(a)|$, a smooth function of $a$. However it appears clearly that this quantity
becomes singular when $\alpha(\tau(a)) = p\cdot\tau(a)=0$: this corresponds to points
of $\Gamma$ such that $T_x\Gamma$ is contained in the kernel of $\alpha$.

Note that we may replace $\alpha$ by $\widetilde{\alpha}(x) = \widetilde{p}\cdot x$,
for $\widetilde{p}\in \mathbb{R}^n$ close to $p$: by choosing $U$ suitably we can show
that the previous computation remains valid for $(\widetilde{\alpha},a)$ close to
$(\alpha,\alpha(x_0))$.
Geometrically $\langle \mu,\widetilde{\alpha}^*\delta_a\rangle$ corresponds to the integral of $\mu$ on
the hyperplane $H_{\widetilde{\alpha},a}$ (more precisely a neighborhood in $H_{\widetilde{\alpha},a}$ of $x_0$),
i.e. the value of the local Radon transform at this hyperplane.

\subsection{The Radon transform of the characteristic function 
of the half-plane}
\label{presentation-Radon}
We go back to the distribution $u$ introduced in Section \ref{half-plane}, i.e.
the characteristic function of the upper half-plane $\Omega$ in $\mathbb{R}^2$.
Any half-line $\{(x,\lambda k)|\lambda \in (0,+\infty)\}$
in the wavefront set of $u$ is characterized by a 
point $x$ and a unit direction $k$. Consider a straight line perpendicular to $k$
and move it along $k$. Then, something should happen to 
the restriction of $u$ to the line when
the line crosses the point $x$.
To be more precise, 
consider a straight line $L_{k,a}$ defined by the equation
$k\cdot x=a$ (the line perpendicular to $k$ that goes through
the point $(a/k_1,0)$ if $k_1\not=0$).  By changing the value of
$a$ we move the line along $k$. 
The integral of $fu$ over the line $L_{k,a}$
is $\mathcal{R}(fu)(k,a)=\int_{L_{k,a}\cap \Omega}f(x)d\ell$ is the value of
the Radon transform of $fu$ at $(k,a)$.
Let us check in this case the result which will be proved
in section~\ref{Radonsect}, i.e. that the wavefront
set of $u$ can be obtained by looking at the
points where the Radon transform
$\mathcal{R}(fu)(k,a)=\int_{L_{k,a}\cap \Omega}f(x)d\ell$ is not a smooth
function of $a$. 
This means here that if a line $L_{k,a}$ is not parallel to the
step, then a small variation of $a$ is smooth,
while it will jump at the step if $L_{k,a}$ is the $x_1$ axis
(see Fig.~\ref{figradon}).
\begin{figure}
\begin{center}
\includegraphics[width=7.0cm]{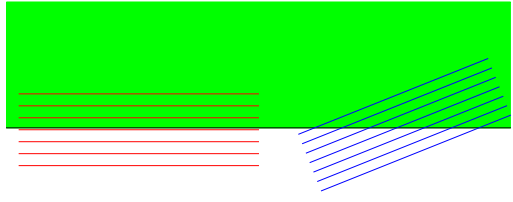}
\caption{ The upper half-plane is green. An integration over
  the blue lines (which are not parallel to the edge) gives a
  smooth function of the distance from the first line.
  An integration over the red lines (parallel to the edge)
  jumps when the line reaches the edge. 
  \label{figradon}}
\end{center}
\end{figure}
To prove this, let $k$ be a unit vector and $v$ a unit vector
perpendicular to $k$. Then the points of the straight line $L_{k,a}$ 
are $x=a k+ t v$ and
\begin{eqnarray*}
\mathcal{R}(fu)(k,a) &=& \int_{-\infty}^\infty f(a k + t v) \theta(a e_2\cdot k+
  t e_2\cdot v) dt,
\end{eqnarray*}
where $e_2$ is the unit vector along the $x_2$ axis.
If we choose an angle $-\pi/2 < \phi \le \pi/2$ such that 
$k=e_1\sin\phi+e_2\cos\phi$ and
$v=e_1\cos\phi-e_2\sin\phi$ (where $e_1$ is the unit vector along the $x_1$ axis) we obtain
$\mathcal{R}(fu)(k,a) = \int_{-\infty}^\infty f(a k  + t v) \theta(a \cos\phi -t \sin\phi)dt$.
We must consider three cases. 
\begin{eqnarray*}
\mathcal{R}(fu)(k,a) &=& \int_{a\cotg\phi }^\infty f(a k  + t v) dt\text{ if } \phi >0, \\
\mathcal{R}(fu)(k,a) &=& \theta(a) \int_{-\infty}^\infty f(a e_2  + t e_1) dt\text{ if } \phi=0\text{ and } a\not=0, \\
\mathcal{R}(fu)(k,a) &=& \int_{-\infty}^{a\cotg\phi } f(a k  + t v) dt\text{ if } \phi <0.
\end{eqnarray*}
We indeed see that $\mathcal{R}(fu)(e_2,a)$ jumps from $0$ for $a=0^-$ to
$\int_{-\infty}^\infty f(te_1) dt$ for $a=0^+$.

\subsection{The wavefront set up to sign and the Radon transform}
\label{Radonsect}
Let us start with the following definition:
\begin{dfn}
 For any distribution $u$ the \emph{wavefront set up to a sign}
 of $u$ is the set
 \[
  \hbox{WF}^{\pm}(u):= \{(x,p)|(x,p)\in \hbox{WF}(u)\hbox{ or }(x,-p)\in \hbox{WF}(u)\}.
 \]
\end{dfn}
This notion is slightly coarser than the wavefront set. However it gives interesting
information about its geometry. Note that $T^*M\setminus\hbox{WF}^{\pm}(u)$ is the set
of \emph{absolutely regular directed points}. These are the points $(x,p)$
such that there exists $f\in \mathcal{D}(\mathbb{R}^n)$ satisfying
$f(x) =1$ and a closed conic 
neighborhood $V\subset \mathbb{R}^n$ 
of $p$ such that $\widehat{fu}$ is fast decreasing on $V\cup (-V)$. The set
$\hbox{WF}^{\pm}(u)$ or equivalently its complementary set
can be characterized by using the Radon transform.

Radon transform is defined by averaging functions on affine subspaces. Here
we use affine hyperplanes of $\mathbb{R}^n$. First consider the case of a
continuous function with compact support $u\in \mathcal{C}^0_c(\mathbb{R}^n)$.
For any $(\nu,s)\in S^{n-1}\times \mathbb{R}$, let $H_{\nu,s}$ be the hyperplane
of equation $\nu\cdot x = s$ and set 
\[
 \mathcal{R}(u)(\nu,s):= \int_{H_{\nu,s}}u(x)d\sigma(x),
\]
where $\sigma$ is the Lebesgue measure on $H_{\nu,s}$. This defines a function $\mathcal{R}(u)$
on $S^{n-1}\times \mathbb{R}$, the \emph{Radon transform} of $u$. This function is linked
to the Fourier transform of $u$ by
\begin{eqnarray*}
 \widehat{u}(p) &=& \int_{\mathbb{R}^n}dx\, e^{ip\cdot x}u(x) =
 \int_{\mathbb{R}}ds\,e^{is|p|}\mathcal{R}(u)\left(p/|p|,s\right)
\\&=&
  \mathcal{F}\left(\mathcal{R}(u)\left(\frac{p}{|p|},\cdot\right)\right)(|p|),
\end{eqnarray*}
hence conversely
\[
 \mathcal{R}(u)(\nu,s) = \frac{1}{2\pi}\int_{\mathbb{R}}dk\,e^{-iks}\widehat{u}(k\nu).
\]
Now consider a distribution $u\in \mathcal{D}'(\mathbb{R}^n)$, let $(x,p)$ be an absolutely
regular directed point of $u$. Let $f\in \mathcal{D}(\mathbb{R}^n)$ such that
$f(x) =1$ and a closed conic neighborhood $V\subset \mathbb{R}^n$ 
of $p$ such that $\widehat{fu}$ is fast decreasing on $V$. 
For any $\nu\in V\cap S^{n-1}$, $k\longmapsto\widehat{fu}(k\nu)$ is a smooth fast decreasing function
of $k\in \mathbb{R}$. We can thus define its inverse Fourier transform and set
\[
 \mathcal{R}(fu)(\nu,s) = \frac{1}{2\pi}\int_{\mathbb{R}}dk\,e^{-iks}\widehat{fu}(k\nu).
\]
Note, for any fixed $\nu$, $s\longmapsto \mathcal{R}(fu)(\nu,s)$ has a compact support
because $f$ has a compact support.
Since $\forall N\in \mathbb{N}$, $\exists C_n>0$ such that $\forall q\in V$,
$|\widehat{fu}(q)|\leq C_N(1+|q|)^N$, it implies that, $\forall m\leq N-2$,
$\forall \nu \in V\cap S^{n-1}$,
$\left|\frac{d^m}{(ds)^m}\mathcal{R}(fu)(\nu,s)\right| \leq C'C_N$, for 
$C' = \frac{1}{2\pi}\int_\mathbb{R} \frac{dk}{(1+|k|)^2}$. Hence $\mathcal{R}(fu)$
is uniformly smooth in $s$ on $(V\cap S^{n-1})\times \mathbb{R}$.

Conversely let $u$ be a distribution and assume that, for some $(x,\nu)\in \mathbb{R}^n\times S^{n-1}$,
there exists $f\in \mathcal{D}(\mathbb{R}^n)$ and a closed neighborhood 
$V\cap S^{n-1}$ of
$\nu$ in $S^{n-1}$ such that we can make sense of the Radon transform $\mathcal{R}(fu)$ of $fu$ on
$(V\cap S^{n-1})\times \mathbb{R}$ (e.g. by proving that there exists a sequence $(fu)_\varepsilon$ of smooth
functions with compact support which converges to $fu$ in $\mathcal{D}'(\mathbb{R}^n)$ and
that the sequence $\mathcal{R}((fu)_\varepsilon)$ converges also in 
$\mathcal{D}'(S^{n-1}\times \mathbb{R})$
to a distribution which we call $\mathcal{R}(fu)$).
Conversely let $u$ be a distribution and observe that, for any 
$f\in \mathcal{D}(\mathbb{R}^n)$ and any closed neighborhood 
$V\cap S^{n-1}$ of $\nu$ in $S^{n-1}$, we can make sense of the Radon 
transform $\mathcal{R}(fu)$ of $fu$ on
$(V\cap S^{n-1})\times \mathbb{R}$, e.g. by noticing that $fu$ is a compactly 
supported distribution, thus $\widehat{fu}$ is real analytic with polynomial 
growth by Paley--Wiener--Schwartz, therefore the restriction 
$\widehat{fu}(k\nu)$ is analytic with polynomial growth in $k\in\mathbb{R}$ 
uniformly in $\nu\in V\cap\mathbb{S}^{n-1}$, 
hence a tempered distribution in $\mathcal{S}^\prime(\mathbb{R})$. 
Its inverse Fourier transform $\mathcal{R}(fu)(\nu,s)$
is thus a tempered distribution in $s$.).
Assume moreover, $\forall N\in \mathbb{N}$, $\exists \Gamma_N>0$ such that $\forall \eta\in V$,
$\left|\frac{d^m}{(ds)^m}\mathcal{R}(fu)(\nu,s)\right| \leq \Gamma_N$. Then, since 
$\forall \eta\in V\cap S^{n-1}$, $s\longmapsto \mathcal{R}(fu)(\eta,s)$
is compactly supported we can define its Fourier transform in $s$ and set
$\widehat{fu}(p) = \int_\mathbb{R}ds\, e^{i|p|s}\mathcal{R}(fu)(p/|p|,s)$,
$\forall p\in V$. It follows
then that $|\widehat{fu}(p)|\leq \Gamma_N|\hbox{supp} \mathcal{R}(fu)(p/|p|,\cdot)||p|^{-N}$
and hence $\widehat{fu}$ is fast decreasing in $V$.
As a conclusion:
\begin{thm}
 Let $u\in \mathcal{D}'(U)$ be a distribution and $(x,k)\in U\times (\mathbb{R}^n\setminus\{0\})$.
 Then $(x,k)$ does not belong to $WF^{\pm}(u)$ iff there exists $f\in \mathcal{D}(U)$
 such that $\mathcal{R}(fu)$ is smooth on a neigborhood of $(k/|k|,k\cdot x/|k|)$
 in $U\times (V\cap S^{n-1})$.
\end{thm}

\section{Oscillatory integrals}\label{oscillatory}
In proposition~\ref{exdisk}, the singular support of the
characteristic distribution is the submanifold $S$.
H\"ormander gives another example of a distribution where the singular support
is a submanifold~\cite[p.~261]{HormanderI}.
This example is important because it exhibits a distribution
defined by an oscillatory integral (as the Wightman
propagator).
\begin{example}
Let $M$ be a smooth submanifold of $\mathbb{R}^n$
defined near a point $x_0\in M$ by 
$\phi_1(x)=\dots = \phi_k(x)=0$ where
$d \phi_1,\dots, d \phi_k$ are linearly independent at $x_0$.
If the function $a\in\calD(\bbR^n)$ has support near $x_0$, 
we define the distribution
$\langle u,f\rangle=(2\pi)^k 
 \int d x a(x) \delta(\phi_1,\dots,\phi_k) f(x)$, 
where
$\delta$ is the delta function in $\mathbb{R}^k$.
This can be rewritten
\begin{eqnarray*}
\langle u,f\rangle &=&  \int_{\bbR^n} dxf(x)
\int_{\bbR^k} d\xi \, a(x) e^{i \phi(x,\xi)},
\end{eqnarray*}
where $\phi(x,\xi)=\sum_{i=1}^k \phi_i(x) \xi_i$ and
$\xi_i\in\mathbb{R}$.
Then $WF(u) = \{(x,-d_x\phi(x,\xi));
\phi_1(x)=\dots=\phi_k(x)=0,x\in\supp a\}$,
where
\begin{eqnarray*}
d_x\phi(x,\xi) &=&
\frac{\partial \phi(x,\xi)}{\partial x_1} d x_1 + \dots
+ \frac{\partial \phi(x,\xi)}{\partial x_n} d x_n.
\end{eqnarray*}
\end{example}
We can use this result to recover the wavefront set of
example~\ref{exdisk} when $n=2$, $\Omega$
is the unit disk and $S$ is the unit circle. 
We have a single function
$\phi_1(x_1,x_2)=x_1^2+x_2^2-1$, so that
$\phi(x,\xi)= (x_1^2+x_2^2-1)\xi$, 
the critical set is given by
$d_\xi\phi(x,\xi)=\phi_1(x_1,x_2)=0$
and
$d_x \phi(x,\xi)= (2 x_1 d x_1 +2 x_2 dx_2)\xi$.
If we switch to polar coordinates, we obtain 
$d_x \phi(x,\xi)= 2 \rho \xi d\rho $,
which is a direction perpendicular to the unit circle
at $x$.  Note that $\xi$ can have both signs, thus
both $d\rho$ and $-d\rho$ belong to the wavefront set.
This example confirms an important characteristics of
the wavefront set. The direction $k$ are not vectors
but covectors. Indeed,
$d_x\phi(x,\xi)$ can be expanded over the
(covector) basis $d x_1,\dots,d x_n$ of $T^*_xM$ 
and not over the vector basis
$\partial_{x_1},\dots,\partial_{x_n}$ of $T_xM$.
To determine the nature of the directions
$k$ in the wavefront set, we can also look at the
way the wavefront set transforms under a smooth mapping
$\mathbb{R}^n\to\mathbb{R}^n$. 
The detailed calculation~\cite[p.~195]{Strichartz-03}
confirms that $k$ are covectors because they transform
covariantly.
This point is important for distributions
on manifolds.

The previous result can be extended to
more generaly oscillatory integrals (in the following we always assume
that the \emph{phase function}
$\phi$ is homogeneous of degree $1$, i.e.
$\phi(x,\lambda \xi) = \lambda \phi(x,\xi)$, $\forall \lambda>0$, 
see~\cite[p.~260]{HormanderI} for details):
\begin{thm}
\label{phasethm}
If a distribution $u$ is defined by an oscillatory integral
\begin{eqnarray*}
u(f) &=&  \int_{\mathbb{R}^n} dx f(x)
\int_{\mathbb{R}^s} d\xi \, a(x,\xi) e^{i \phi(x,\xi)} d \xi,
\end{eqnarray*}
where $\phi$ is a phase function 
and $a$ an asymptotic symbol, then
$WF(u)\subset\{(x;-d_x\phi(x,\xi))\,|\,
d_\xi\phi(x,\xi)=0\}$.
\end{thm}
We refer to the literature for a precise definition of phase functions 
and asymptotic symbols~\cite[p.~236]{HormanderI}\cite[p.~99]{ReedSimonII}.
We can give a hand-waving argument to understand the origin of
this wavefront set. The Fourier transform of 
$u$ is given by 
$\int dx \int d\xi a(x,\xi) e^{i\phi(x,\xi)+i k\cdot x}$.
By using the stationary-phase method, we see that 
the directions of slow decrease are the directions 
where the phase $\phi(x,\xi)+k\cdot x$ is critical with respect to $(x,\xi)$. They
are determined by the equations
$k+d_x\phi(x,\xi)=0$ and $d_\xi\phi(x,\xi)=0$.

\subsection{The Wightman propagator}
With the help of theorem~\ref{phasethm}
we can calculate the wavefront set of
a fundamental distribution of quantum field theory:
the Wightman propagator in Minkowski spacetime~\cite[p.~66]{ReedSimonII}
\begin{eqnarray*}
W_2(f,g) &=& \int_{\mathbb{R}^3\times \mathbb{R}^3} dx d y 
\langle 0|\varphi(x) \varphi(y) |0\rangle
f(x) g(y)
\\&=&
-i \int_{\mathbb{R}^3\times \mathbb{R}^3} dx dy 
\Delta_+(x-y)
f(x) g(y),
\end{eqnarray*}
where~\cite[p.~70]{ReedSimonII}
\begin{eqnarray*}
\Delta_+(x) &=&
\frac{i}{2(2\pi)^3}
\int_{\mathbb{R}^3} \frac{d\bfk}{k_0}
e^{-i k\cdot x},
\end{eqnarray*}
with $\bfk=(k_1,k_2,k_3)$, $k_0=\sqrt{|\bfk|^2+m^2}$,
$k = (k_\mu) = (k_0,\bfk)$ and $k\cdot x = \sum_{\mu=0}^3k_\mu x^\mu$.

The analytic form of 
$\Delta_+(x)$ 
is given by Scharf~\cite[p.~90]{Scharf}.
We can write $\Delta_+(f)$ in the oscillatory integral
form of Theorem~\ref{phasethm} by setting, 
for $\xi\in \mathbb{R}^3$,~\cite[p.~100]{ReedSimonII}
\begin{eqnarray*}
\phi(x,\xi) &=& -x^0 |\xi| 
- x^j \xi_j,\quad
a(x,\xi) = \frac{e^{-ix^0(\omega_\xi
-|\xi|)}}{\omega_\xi},
\end{eqnarray*}
where 
$|\xi|=\sqrt{(\xi_1)^2+(\xi_2)^2+(\xi_3)^2}$,
$\omega_\xi=\sqrt{|\xi|^2+m^2}$
and
$ x^j \xi_j=\sum_{j=1}^3  x^j \xi_j$.
To prove this, just write
\begin{eqnarray*}
-ik\cdot x &=&  -i(x^j k_j + x^0k_0)
=i(-x^0 |\bfk|-x^j k_j)
-ix^0 (k_0 -|\bfk|),
\end{eqnarray*}
and replace $\bfk$ by $\xi$.
The modification of the phase is necessary to
make $a(x,\xi)$ an asymptotic symbol.
We can now calculate the wavefront set of the 
Wightman propagator~\cite[p.~106]{ReedSimonII}
\begin{prop}
\label{exDeltap}
The wavefront set of $\Delta_+$
is
$\WF(\Delta_+)=S_0 \cup S_+ \cup S_-$, where
\begin{eqnarray*}
S_0 &=& 
\{ (0;|\bfk|,\bfk)\,|\, \bfk\in(\mathbb{R}^3\backslash\{0\})\},\\
S_{\pm} &=& 
\{ (\pm|\bfx|,\bfx;\lambda|\bfx|,\mp\lambda \bfx) \,|\,
\bfx\in(\mathbb{R}^3\backslash\{0\}),\lambda>0\}.
\end{eqnarray*}
More compactly~\cite[p.~118]{Strohmaier},
$\WF(\Delta_+)=\{(x;k)\telque k_0=|\bfk|, 
x^0=\lambda k_0, x^i= -\lambda k_i, \lambda\in\mathbb{R}\}$.
\end{prop}
\begin{figure}
\begin{center}
\includegraphics[width=5.0cm]{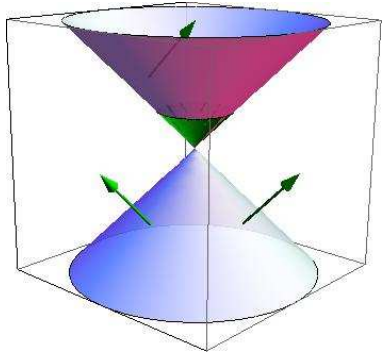}
\caption{Wavefront set of $\Delta_+$: the
wavefront set at the origin is an upper cone.
Note that, in this figure, three different spaces
are identified: the configuration space $\bbR^3$,
the tangent space and the cotangent space 
over each point of the configuration point.
The tangent and cotangent spaces are identified through
the Euclidian metric. This implies that the covectors in
$\WF(u)$ are perpendicular to the tangent planes.
  \label{figdeltaplus}}
\end{center}
\end{figure}
The advantage of the physical convention for the Fourier
transform is that positive energies correspond to
$k_0>0$. 
The wavefront set of $\Delta_+$ for curved (globally hyperbolic)
spacetime is given by Strohmaier~\cite{Strohmaier}.

\begin{proof}
According to Theorem \ref{phasethm}, we first
calculate the set of critical points $\{d_\xi\phi=0\}$
for $\phi(x,\xi)=-x^0\vert\xi\vert-x^i\xi_i$. 
We find 
$ -x^0\left(\sum_{i=1}^3\frac{\xi_id\xi_i}{\vert\xi\vert}\right)
-\sum_{i=1}^3x^id\xi_i=0$, which implies
$x^i=-x^0\frac{\xi_i}{\vert\xi\vert}$ and thus
$x^0=\lambda \vert\xi\vert$ and $x^i=-\lambda\xi_i$ for 
$\lambda=\frac{x^0}{\vert\xi\vert}.$
Conversely, if we plug $x^0=\lambda \vert\xi\vert, x^i=-\lambda\xi_i $ in $d_\xi\phi$
for any $\lambda\in\mathbb{R}$,
we find $d_\xi\phi=\sum_{i=1}^3(\lambda\xi_i-\lambda\xi_i)d\xi_i=0$.
Then Theorem \ref{phasethm} claims that $WF(\Delta_+)$
is a subset of
$\{(x;d_x\phi) \telque d_\xi\phi(x;\xi)=0\}$:
\begin{eqnarray*}
WF(\Delta_+) &\subset &
\{(x^0,\bfx;\vert\xi\vert,\xi_i ) |  x^0=\lambda \vert\xi\vert, x^i=-\lambda\xi_i ,\lambda\in\mathbb{R} \}
\\&\subset&
\{(x^0,\bfx;k_0,k_i ) |  k_0=\vert\bfk\vert, x^0=\lambda k_0, x^i=-\lambda k_i ,\lambda\in\mathbb{R} \}.
\end{eqnarray*}
We leave to the
reader the proof of the decomposition 
$\{(x^0,\bfx;k_0,k_i ) |   k_0=\vert\bfk\vert, x^0=\lambda k_0, x^i=-\lambda k_i ,\lambda\in\mathbb{R} \}= S_0\cup S_+\cup S_-$.

Note that theorem~\ref{phasethm} states only that
$\WF(\Delta_+)\subset S_0\cup S_+\cup S_-$. 
We refer to the literature to show that $\subset$
can be replaced by $=$~\cite[p.~107]{ReedSimonII}.
The singular support of $\Delta_+$ is the light cone
$x^0=\pm |\bfx|$, the cotangent vectors $k$
are light-like, have positive energy $k_0$
and are perpendicular to $x$.
\end{proof}

\subsection{The Feynman propagator}
\label{Feynpropsect}
\begin{prop}
\label{exDeltaF}
The Feynman propagator 
$\Delta_F(x)=\theta(x^0) \Delta_+(x)
+ \theta(-x^0) \Delta_+(-x)$ exists 
and its
wavefront set is
$\WF(\Delta_F)=D^* \cup C_F$, where
$D^*=\{(0;k)\,|\, k\not=0\}$ is the wavefront set of
the Dirac delta function and
$C_F=\{(x;k)\,|\, (x^0)^2 - |\bfx|^2=0, x^0\not=0, 
k_0= \lambda x^0,k_i = - \lambda x^i, \lambda>0\}$.
\end{prop}
\begin{figure}
\begin{center}
\includegraphics[width=5.0cm]{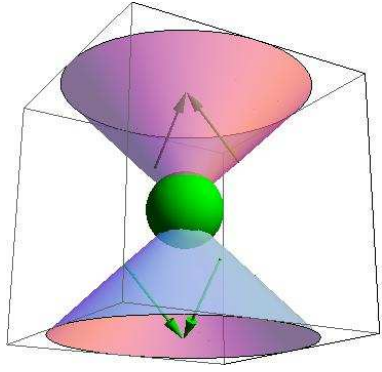}
\caption{Wavefront set of $\Delta_F$, the
wavefront set at the origin is a ball.
  \label{figfeynman}}
\end{center}
\end{figure}
\begin{proof}
$\theta(x^0) \Delta_+(x)$ is a product of distributions,
we must first show that it exists. As a distribution
in $\mathbb{R}^4$, $\theta(x^0)$ is defined by
$\theta(x^0)(f)=\int_{x^0\geqslant 0} f(x) dx$. Therefore,
it is the tensor product of the Heaviside 
distribution in the variable $x^0$ by the unit
distribution in the variables $x^1,\dots,x^3$:
$\theta(x^0)=\theta\otimes 1$. The distribution 1
is smooth and its wavefront set is empty. 
Thus, by property (i) of section~\ref{propWF},
we have
$\WF(\theta(x^0))\subset\{(0,\bfx;\pm\lambda,0)\,|\,
\bfx\in\mathbb{R}^3, \lambda>0\}$.
In fact, the inclusion can be replaced by
an equal sign~\cite[p.~108]{ReedSimonII}.
By theorem~\ref{disprodthm}, we see that
the product $\theta(x^0)\Delta_+(x)$ exists.
Indeed, $\singsupp \theta(x^0) \cap \singsupp \Delta_+
= \{0\}$ and, at $x=0$, the
allowed cotangent vectors are
$k=(\pm\lambda,0)$ with $\lambda>0$ for $\theta(x^0)$ and
$q=(|\bfk|,\bfk)$ with $\bfk\not=0$ for $\Delta_+$.
Thus, $q+k\not=0$ and the product exists. 
A similar calculation for $\theta(-x^0)\Delta_F(-x)$ shows
that $\Delta_F$ is a well defined distribution on $\bbR^4$.
However, the estimate of the wavefront set given by the Product
Theorem is not precise enough because of the
contribution of $\WF(\theta)$.
To calculate the wavefront set of
$\Delta_F$, we use the causality method
of Bogoliubov. Let 
$x=(x^0,\bfx)\in\mathbb{R}^4\setminus\{0\}$.
If $x^0>0$ then there 
is a neighborhood $U$ of $x$ 
such that $\forall y\in U, y^0>0$.
Therefore, 
$\Delta_F|_U=\theta(x^0)\Delta_+(x)|_U=\Delta_+(x)|_U$
thus $WF(\Delta_F|_U)=WF(\Delta_+|_U)=S_+|_U $ 
and by definition of $S_+$:
$$
WF(\Delta_F|_U)=
\{(x;k)\telque x^0=|\bfx|, x^0=\lambda k_0,
x^i=-\lambda k_i, \lambda>0, x\in U \}.
$$
If $x^0<0$ then 
there is a neighborhood $U$ of $x$ such that
$\Delta_F|_U=\theta(-x^0)\Delta_+(-x)|_U=\Delta_+(-x)|_U$.
Thus 
\begin{eqnarray*}
WF(\Delta_F|_U) &=& \{(x;k) \telque (-x;-k)\in S_+, x\in U\}
\\&=&
\{(x;k)\telque x^0=-|\bfx|,
x^0=\lambda k_0, x^i=-\lambda k_i, \lambda>0, x\in U \}.
\end{eqnarray*}
If $x^0=0$, then $x$ is space-like because $x\not=0$.
Thus, there exists some orthochronous Lorentz
transformation
$R\in SO^\uparrow(1,3)$ such that $(Rx)^{0}>0$. From the definition of $\Delta_F$ we deduce that
$\Delta_F(Rx)=\theta((Rx)^0) \Delta_+(Rx) + \theta(-(Rx)^0) \Delta_+(-Rx)$. Since
$\Delta_F$ and $\Delta_+$ are invariant by orthochronous Lorentz transformations, this
implies $\Delta_F(x)=\theta((Rx)^0) \Delta_+(x) + \theta(-(Rx)^0) \Delta_+(-x)$. Hence
we recover the case $x^0>0$ and $\Delta_F$ is smooth on a neighborhood of $x$ 
because $x$ is not light-like.
This gives us $\WF(\Delta_F)|_{x\not=0}=C_F$. 
To complete the proof of the proposition, recall that 
$(\Box+m^2)\Delta_F=-i\delta$~\cite[p.~124]{Itzykson}. Thus
property (h) of section~\ref{propWF} implies
$\WF(\delta)=D^*\subset\WF(\Delta_F)$. Since no wavefront
set at $x=0$ can be larger than $D^*$, we obtain
$\WF(\Delta_F)|_{x=0}=D^*$ and the proposition is proved.
\end{proof}

The calculation of $\WF(\Delta_F)$ was first made by
Duistermaat and H\"ormander~\cite{Duistermaat-72} 
after discussion with Wightman.
The analytic expression for the Feynman propagator
in position space is given by Zhang et al.~\cite{Zhang-10}.
The wavefront set of the advanced and retarded
solutions to the wave equation is
calculated in ~\cite{Duistermaat-72} and \cite[p.~115]{Strohmaier}.

\section{Properties of the wavefront set}
\label{propWF}
We now give without proof a number of properties
of the wavefront set.
Let $u$ and $v\in \calD'(\mathbb{R}^n)$. Then
\begin{itemize}
\item[(a)] $\WF(u)$ is a closed subset of
$\mathbb{R}^n\times
(\mathbb{R}^n\backslash\{0\})$~\cite[p.~92]{ReedSimonII}.
\item[(b)] For each $x\in \mathbb{R}^n$,
$\WF_x(u)=\{k \telque (x,k)\in \WF(u)\}$ is a cone,
i.e. $k\in\WF_x(u)$ and $\lambda>0$ implies
$\lambda k\in\WF_x(u)$~\cite[p.~92]{ReedSimonII}.
\item[(c)] $\WF(u+v)\subset \WF(u)\cup\WF(v)$~\cite[p.~92]{ReedSimonII}.
\item[(d)] $\singsupp u=\{x\telque WF_x(u)
\not=\emptyset\}$~\cite[p.~93]{ReedSimonII}.
\item[(e)] If $u$ is a tempered distribution and $\hat{u}$ has support
in a closed cone $C$, then for each $x$,
$\WF_x(u)\subset C$~\cite[p.~93]{ReedSimonII}.
\item[(f)] Let $U\subset \mathbb{R}^m$ and $V\subset \mathbb{R}^n$
be two open sets. For any smooth ($C^\infty$) map $f:U\longrightarrow V$
we define
\[
 N_f:= \{(f(x), k)\in V\times \mathbb{R}^n\telque
  k\circ df_x =0\},
\]
where $ k\circ df_x := ( k_idy^i)\circ df_x:= 
 k_idf^i_x$.
Consider the \emph{pull-back operator} 
$u\longmapsto f^*u:= u\circ f$ defined
on smooth maps $u$ on $V$.
Then it is possible to extend this operator to
 the space of distributions $u\in \mathcal{D}'(V)$ 
which satisfy $N_f\cap WF(u) = \emptyset$ in an unique way (if we furthermore require
some continuity assumptions, see \cite[Thm 8.2.4]{HormanderI}).
Moreover the wavefront set of $f^*u$ is contained in the set
\[
 f^*WF(u):= \{(x, k\circ df_x)|(f(x),k)\in WF(u)\}.
\]
\cite[Thm 8.2.4]{HormanderI} (beware that, in the definition of the inverse image of a distribution
by a diffeomorphism in \cite[p.~93]{ReedSimonII}, the expression for the
wavefront set of $f^*u$ is not correct.)
\item[(g)] If $u\in \calD'(U)$ and $f\in C^\infty(U)$, then
$\WF(fu)\subset \{\supp f\times (\mathbb{R}^n\backslash\{0\})\}
\cap \WF(u)$~\cite[p.~344]{Wagschal-11}.
\item[(h)] If $u\in \calD'(U)$ and 
$P$ is a partial differential operator with smooth
coefficients, then
$\WF(P u)\subset\WF(u)$~\cite[p.~256]{HormanderI}.
\item[(i)] If $u\in \calD'(U)$ and $v\in \calD'(V)$, then
$\WF(u\otimes v) \subset \big(\WF(u)\times \WF(v)\big)
\cup \big((\supp u\times \{0\})\times \WF(v)\big)
\cup \big(\WF(u)\times (\supp v\times 
\{0\})\big)$~\cite[p.~267]{HormanderI}.
\item[(j)] If $u\in \calD'(U\times U)$ is such that
(formally) $u(x,y)=v(x-y)$ for some
$v\in\calD'(V)$, then 
$\WF(u)=\{(x,y;k,-k)\telque (x-y;k)\in \WF(v)\}$~\cite[p.~118]{Strohmaier}
and \cite[p.~270]{HormanderI}.
\end{itemize}
As an application of the pull-back theorem, we
calculate the wavefront set of $\Delta_F$ for
a massless particle, whose analytic expression 
is~\cite[p.~133]{Itzykson}
\begin{eqnarray*}
\Delta_F(x)=\frac{1}{4\pi^2}\frac{1}{x^2-i0},
\end{eqnarray*}
where $x^2=(x^0)^2-|\bfx|^2$.
We first prove that
this distribution
is well defined
on $\mathbb{R}^4\setminus\{0\}$.
So $\Delta_F$
is just the
pull--back
of $(2\pi)^{-2}\left(t-i0\right)^{-1}$
by the $C^\infty$ map.
\begin{eqnarray}
f:x\in\mathbb{R}^4\setminus\{0\}\longmapsto 
(x^0)^2-\vert\bfx\vert^2\in\mathbb{R}.
\end{eqnarray}
Indeed, this map is smooth and
$N_f=\{(f(x),k)\telque 2 k (x^0dx^0-x^idx^i)=0\}$.
We know that $\WF(1/(t-i0^+))=\{(0;k)\telque k>0\}$.
Thus, the condition $N_f\cap WF(u)=\emptyset$
implies $x\not=0$ and  $\Delta_F$
is therefore well defined 
in $\mathcal{D}^\prime(\mathbb{R}^4\setminus\{0\})$.
Furthermore, by property (f) 
$WF(\Delta_F|_{x\not=0})$ is included in
the pull--back of $\WF(1/(t-i0^+))$ 
by $f$. 
We obtain:
\begin{eqnarray*}
WF(\Delta_F|_{x\not=0}) &\subset & 
\{(x;\lambda\circ df) \telque (f(x);\lambda)\in 
WF\left(t-i0\right)^{-1}\}.
\end{eqnarray*}
Therefore,
$WF(\Delta_F|_{x\not=0})=\{(x;k) | f(x)=0, k=\lambda
df(x)\subset\lambda(x^0,-\bfx),\lambda>0,x\neq 0  \}$.
To conclude, observe that
$\Delta_F$ is a homogeneous distribution, 
therefore by a theorem of H\"ormander
(\cite[Thm 3.2.4]{HormanderI}), it
admits an extension in 
$\mathcal{D}^\prime(\mathbb{R}^4)$.
The wavefront set of $\Delta_F$ at $x=0$ is calculated as
in the proof of Prop.~\ref{exDeltaF} by using
$\Box\Delta_F=-i\delta$ and we
recover  Prop.~\ref{exDeltaF} for $m=0$.

\section{The many faces of the wavefront set}
In this section we give several definitions
of the wavefront set.
Each of them can be useful in specific contexts.

\subsection{The frequency set}
It is possible to define the wavefront set in
terms of the \emph{frequency set} of distributions $u$,
denoted by $\Sigma(u)$~\cite[p.~254]{HormanderI},
which is the projection of the wavefront set of $u$ 
on the momentum (i.e. cotangent) space:
\begin{dfn}
\label{dfnHor}
Let  $u\in\calE'(\bbR^n)$, we define $\Sigma(u)$ to
be the closed cone in $\mathbb{R}^n\backslash\{0\}$
having no conic neighborhood $V$ such that,
$|\hat{u}(k)|\le C_N (1+|k|)^{-N}$
for $k\in V$ and for all $N=1,2,\dots$.
\end{dfn}
Friedlander and Joshi define the \emph{frequency set}
$\Sigma(u)$ by
\begin{dfn}
\label{dfnFried}
Let $u\in \calE'(\bbR^n)$, then the direction $k_0$ is not in
$\Sigma(u)\subset \mathbb{R}^n\backslash\{0\}$ iff there is
a conic neighborhood $V$ of $k_0$ such that, for all $N$,
there is a $C'_N$ such that
$|\hat{u}(k)|\le C'_N \langle k\rangle^{-N}$,
for all $k$ in $V$, where
$\langle k\rangle=(1+|k|^2)^{1/2}$.
\end{dfn}
Duistermaat (implicitly) proposed a third definition
\begin{dfn}
\label{dfnDuist}
Let $u\in \calE'(\bbR^n)$, then the direction $k_0$ is not in
$\Sigma(u)\subset \mathbb{R}^n\backslash\{0\}$ iff there is
a neighborhood $W$ of $k_0$ such that, for all $N$,
there is a constant $D_N$ such that
$|\hat{u}(\tau k)| \le D_N \tau^{-N}$ for $\tau\to\infty$ uniformly
in $k\in W$.
\end{dfn}
The proof of the equivalence of these definitions
is left to the reader.

\subsection{Several definitions of the wavefront set}
The frequency set is used in several definitions
of the wavefront set.
According to H\"ormander~\cite[p.~254]{HormanderI}
\begin{dfn}
\label{defWFHor}
Let $U$ be an open set of $\mathbb{R}^n$,
$u\in\calD'(U)$ and
$\Sigma_x(u) = \bigcap_{\phi} \Sigma(\phi u)$,
where $\phi$ runs over all elements of $\calD(U)$ such
that $\phi(x)\not=0$.
The wavefront set of $u$
is the closed subset of
$U\times (\mathbb{R}^n\backslash\{0\})$ defined by
\begin{eqnarray*}
WF(u)=\{(x;k)\in U\times (\mathbb{R}^n\backslash\{0\})
\telque k\in\Sigma_x(u)\}.
\end{eqnarray*}
\end{dfn}
For Duistermaat~\cite[p.~16]{Duistermaat} the wavefront set is:
\begin{dfn}
\label{defWFDuist}
If $u\in\calD'(U)$, then $WF(u)$ is defined as
the complement in
$U\times (\mathbb{R}^n\backslash\{0\})$  of the
collection of all $(x_0,k_0)\in
(\mathbb{R}^n\backslash\{0\})$ such that for
some neighborhood $U$ of $x_0$,
$V$ of $k_0$ we have for each $\phi\in \calD(U)$
and each $N$:
$\widehat{\phi u}(\tau k)=O(\tau^{-N})$ for
$\tau\to\infty$, uniformly in $k\in V$.
\end{dfn}
An equivalent definition was used by Chazarain and
Piriou~\cite[p.~501]{Chazarain}, who use the
name \emph{singular spectrum} but the notation
$\WF u$.

For Friedlander and Joshi~\cite[p.~145]{Friedlander}
(after correction of a misprint)
and Strichartz~\cite[p.~191]{Strichartz-03}
\begin{dfn}
\label{defWFFried}
Let $Y$ be an open set of $\mathbb{R}^n$
and $u\in\calD'(U)$, then we shall say that
$(x_0,k_0)\in U\times (\mathbb{R}^n\backslash\{0\})$
is not in $\WF(u)$ iff there exists
$\phi\in \calD(U)$ such that
$\varphi(x_0)\not=0$ and $k_0\notin \Sigma(\phi u)$.
\end{dfn}
For Eskin~\cite[p.~58]{Eskin}
\begin{dfn}
\label{defWFEskin}
Let $U$ be an open set of $\mathbb{R}^n$
and $u\in\calD'(U)$, then we shall say that
$(x_0,k_0)\in U\times (\mathbb{R}^n\backslash\{0\})$
is not in $\WF(u)$ iff there exists
$\phi\in \calD(U)$ such that
$\varphi(x_0)\not=0$ and
$|\widehat{\phi u}(k)|\le C_N (1+|k])^{-N}$
for all $N$ and all $k\not=0$ satisfying
$|\frac{k}{|k|}-\frac{k_0}{|k_0|}| < \delta$ for some $\delta >0$.
\end{dfn}
The proof of the equivalence of these definitions
is left to the reader.

\subsection{More definitions of the wavefront set}
In this section we gather alternative definitions
of the wavefront set, which show that the wavefront set is the single
solution of many different problems.

\subsubsection{Coordinate invariant definition}

A coordinate invariant definition of the wavefront
set was given by Duistermaat~\cite[p.~16]{Duistermaat},
following a first attempt by Gabor~\cite{Gabor-72}.
We consider a smooth $n$-dimensional manifold $M$,
its cotangent bundle $T^*M$ and the zero section
$Z$ of $T^*M$ (i.e. $Z=\{(x;k)\in T^*M\telque k=0\}$).
Then, the wavefront set of a distribution
$u\in\calD'(M)$ is a closed
conic subset of $\dotT^*M=T^*M\backslash Z$:
\begin{dfn}
If $M$ is a smooth $n$-dimensional manifold,
$u\in\calD'(M)$ and $(x_0;k_0)\in \dotT^*M$,
then $(x_0;k_0)\notin\WF(u)$
iff, for any smooth function
$\psi: M\times \mathbb{R}^p\to \mathbb{R}$,
with $d_x\psi(x_0,a_0)=k_0$, there are
open neighborhoods $U$ of $x_0$ and $A$ of $a_0$
such that, for any $\phi\in\calD(U)$ we have
for all $N\ge 1$:
$\langle u,e^{i\tau\psi(\cdot,a)}\phi\rangle=
O(\tau^{-N})$ for $\tau\to\infty$,
uniformly in $a\in A$.
\end{dfn}
This definition is surprisingly general because
the phase function $\psi$ is only required to
be smooth and to satisfy $d_x\psi(x_0,a_0)=k_0$.
The usual definition of the wavefront set is recovered
by choosing $A=\bbR^n$, $a=k$ and
$\psi(x,a)=k\cdot x$. In the coordinate invariant definition,
the open set $A$ is used to parametrize the 
covectors $k_0$ of the wavefront set
but its dimension $p$ is not necessarily equal to $n$.
Still, this general definition is equivalent to
the standard one (see Refs.~\cite[p.~17]{Duistermaat},
\cite[p.~542]{Chazarain} and \cite{Viet-wf2}).

\subsubsection{Pseudo-differential operators}
The original definition of the wavefront set was given by
H\"ormander in terms of pseudo-differential 
operators~\cite{Hormander-71},\cite[p.~89]{HormanderIII}:
\begin{eqnarray*}
\WF(u) &=& \bigcap\{\mathrm{char} P\telque Pu\in C^\infty(\bbR^n)\},
\end{eqnarray*}
where $P$ runs over the pseudo-differential operators of all 
orders~\cite[p.~85]{HormanderIII}. If $p_m(x,k)$ is the
principal symbol of $P$, then
$\mathrm{char} P=\{(x,k)\in \bbR^n\times(\bbR^n\backslash\{0\})
  \telque p_m(x,k)=0\}$ is the set of characteristic points
of $P$~\cite[p.~87]{HormanderIII}.
A proof of the equivalence with the other definitions can
be found in ~\cite[p.~307]{Folland}
(see also \cite[p.~78]{Grigis}).

\subsubsection{Wavelets and Co}
In the usual definitions of the wavefront set, the
distribution $u$ is multiplied by a large family
of test functions $f$ and the product is Fourier
transformed. It is in fact possible to use a single
function $f$ and to scale it.
More precisely, let $f$ be an even Schwartz function that does not
vanish at zero and, for any $\alpha$ with $0<\alpha<1$,
form the family of Schwartz functions
$f_t(y)=t^{\alpha n/2}f\big(t^\alpha(y-x)\big)$ for $t>0$.
Then, $(x,k)$ is not in the wavefront set of the tempered distribution
$u$ iff there exists an open
subset $U$ of $\bbR^n$ such that 
$\widehat{uf_t}(tq)$ is fast decreasing in the variable
$t>0$ uniformly in $q\in U$~\cite[p.~159]{Folland-89}.
This definition was first proposed by C{\'o}rdoba and Fefferman
for $\alpha=1/2$ and $f$ a Gaussian function~\cite{Cordoba-78}.
It is then similar to the FBI-transform
(see Ref.~\cite{Martinez} for a nice presentation and
Ref.~\cite{Wunsch-00} for a geometric version).

Although wavelets cannot be used to measure the wavefront set
because they are isotropic, some variants of them,
known as curvelets~\cite{Candes-05} or
shearlets~\cite{Grohs-12} provide an interesting
resolution of the wavefront set.

\section{Conclusion}
We have presented a review of the various guises of the wavefront
set. These different points of view should help grasp the meaning
of this concept. We also proposed two new descriptions of
the wavefront set of a characteristic distribution. Physically,
 we saw that the wavefront set is
related to the fact that, in some directions, destructive interferences
in Fourier space become weaker than for smooth functions. 
The wavefront set also describes
the directions along which the singularities of the distribution
propagate.
We hope that we have convinced the reader that the wavefront
set is a subtle but natural object. Its use should
not be limited to quantum field theory or many-body physics
because, as stressed by Martinez, it is also related
to the semi-classical limit~\cite[p.~134]{Martinez}.

It is ironic that, although the standard wavefront set is sufficient
to build a quantum theory of gauge fields and gravitation, 
it is not enough to describe the optics of crystals (in particular
the conical refraction). Higher order wavefront
sets were proposed~\cite{Liess} to solve that problem.

Finally, note that we have restricted our discussion
to the classical wavefront set.
Many variations have been devised:
analytic wavefront set (see~\cite{Strohmaier-02} and
\cite{Harris-13} for
a recent comparison of various definitions), homogeneous
wavefront set~\cite{Nakamura-05},
Gabor wavefront set~\cite{Rodino-14},
global wavefront 
set~\cite{Coriasco-13,Cappiello}, etc.

\section{Acknowledgements}
We thank Etienne Balan for useful critical remarks
and Andr\'e Martinez for stressing the semi-classical aspect
of the wavefront set.

\end{document}